\documentclass[11pt]{article} 
\usepackage{times}
\usepackage{amsmath,amsfonts}
\usepackage{epsfig,amssymb,amstext,xspace,theorem}
\usepackage{graphicx,color} 
\usepackage{mathrsfs}

\newcommand{\np}{{\em NP}\xspace}

\newcommand{\ztime}{{\em ZTIME}\xspace}
\newcommand{\E}[2][{}]{\ensuremath{\mathrm{E}_{#1}\bigl[#2\bigr]}}   

\newtheorem{theorem}{Theorem}[section]
\newtheorem{lemma}[theorem]{Lemma}

\newtheorem{claim}[theorem]{Claim}

\newtheorem{corollary}[theorem]{Corollary}

{\theorembodyfont{\rmfamily} 
\newtheorem{algorithm}{Algorithm}}
\newcommand{\lp}[1]{\ensuremath{\eqref{mlufllp}_{#1}}}

\def\blksquare{\rule{2mm}{2mm}}
\def\qedsymbol{\blksquare}
\newcommand{\bg}[1]{\medskip\noindent{\bf #1}}
\newcommand{\ed}{{\hfill\qedsymbol}\medskip}
\newenvironment{proof}{\bg{Proof : }}{\ed}
\newenvironment{proofof}[1]{\bg{Proof of #1 : }}{\ed}

\newenvironment{labellist}[1]{\begin{list}{{#1}\arabic{enumi}.}{\usecounter{enumi}
\addtolength{\leftmargin}{-1ex}}}{\end{list}}

\newcommand{\A}{\ensuremath{\mathcal{A}}}
\newcommand{\C}{\ensuremath{\mathcal{C}}}

\newcommand{\I}{\ensuremath{\mathcal I}}
\newcommand{\F}{\ensuremath{\mathcal F}}
\newcommand{\D}{\ensuremath{\mathcal D}}
\newcommand{\T}{\ensuremath{\mathcal T}}
\newcommand{\Nc}{\ensuremath{\mathcal N}}

\newcommand{\Pc}{\ensuremath{\mathcal P}}
\newcommand{\Qc}{\ensuremath{\mathcal Q}}
\newcommand{\OPT}{\ensuremath{\mathit{OPT}}}

\newcommand{\frall}{\ensuremath{\text{ for all }}}

\newcommand{\sm}{\ensuremath{\setminus}}
\newcommand{\es}{\ensuremath{\emptyset}}

\newcommand{\assign}{\ensuremath{\leftarrow}}

\newcommand{\ceil}[1]{\ensuremath{\left\lceil#1\right\rceil}}
\newcommand{\floor}[1]{\ensuremath{\left\lfloor#1\right\rfloor}}
\newcommand{\poly}{\operatorname{\mathsf{poly}}}
\newcommand{\polylog}{\operatorname{polylog}}

\newcommand{\e}{\ensuremath{\epsilon}}

\newcommand{\gm}{\ensuremath{\gamma}}

\newcommand{\sse}{\subseteq}

\newcommand{\ignore}[1]{}

\def\z{\mathfrak{z}}

\def\Time{\mathsf{T}}
\def\T{{\mathcal T}}

\def\br#1{{{(#1)}}}

\newcommand{\ufl}{{\small \textsf{UFL}}\xspace}
\newcommand{\ml}{{\small \textsf{ML}}\xspace}
\newcommand{\tsp}{{\small \textsf{TSP}}\xspace}
\newcommand{\mlufl}{{\small \textsf{MLUFL}}\xspace}
\newcommand{\mssc}{{\small \textsf{MSSC}}\xspace}
\newcommand{\mgl}{{\small \textsf{MGL}}\xspace}
\newcommand{\gst}{{\small \textsf{GST}}\xspace}
\newcommand{\zfc}{{\small \textsf{ZFC}}\xspace}
\newcommand{\mufl}{\ensuremath{\mathsf{UFL}}}
\newcommand{\mzfc}{\ensuremath{\mathsf{ZFC}}}
\newcommand{\kmed}{\ensuremath{k\mathsf{Med}}}
\newcommand{\fcost}{\ensuremath{F^*}}
\newcommand{\ccost}{\ensuremath{C^*}}
\newcommand{\lcost}{\ensuremath{L^*}}
\newcommand{\iopt}{\ensuremath{O^*}}
\newcommand{\tx}{\ensuremath{\tilde x}}
\newcommand{\ty}{\ensuremath{\tilde y}}

\newcommand{\hx}{\ensuremath{\hat x}}
\newcommand{\hy}{\ensuremath{\hat y}}

\newcommand{\bx}{\ensuremath{\bar x}}

\newcommand{\ld}{\ensuremath{\lambda}}
\newcommand{\kp}{\ensuremath{\kappa}}
\newcommand{\al}{\ensuremath{\alpha}}
\newcommand{\tht}{\ensuremath{\theta}}

\newcommand{\dt}{\ensuremath{\delta}}
\newcommand{\sg}{\ensuremath{\sigma}}

\newcommand{\bC}{\ensuremath{C^*}}
\newcommand{\bL}{\ensuremath{L^*}}
\newcommand{\nt}{\ensuremath{S}}
\newcommand{\nbr}{\ensuremath{\mathsf{nbr}}}
\newcommand{\TS}{\ensuremath{\mathsf{TS}}}
\newcommand{\bbL}{\ensuremath{\overline L}}
\newcommand{\bdE}{\ensuremath{\mathbf{E}}}
\newcommand{\bdF}{\ensuremath{\mathbf{F}}}
\newcommand{\bdD}{\ensuremath{\mathbf{D}}}

\newcommand{\LC}{\ensuremath{\mathsf{Lcost}}}
\newcommand{\Lat}{\ensuremath{\mathsf{Lat}}}
\newcommand{\Lf}{\ensuremath{\mathcal{L}}}

\newcommand{\Tour}{\ensuremath{\mathsf{Tour}}}
\newcommand{\Pcol}{\ensuremath{\Pc}}
\newcommand{\Tcol}{\ensuremath{\T}}
\newcommand{\Pfeas}{\ensuremath{\mathsf{Pfeas}}}
\newcommand{\Tfeas}{\ensuremath{\mathsf{Tfeas}}}

\title{Facility Location with Client Latencies: Linear-Programming based Techniques for
Minimum-Latency Problems} 
\author{
         Deeparnab Chakrabarty\thanks{{\tt deepc@seas.upenn.edu.} 
	 Dept. of Computer and Information Science, Univ. of Pennsylvania, Philadelphia,
	 PA 19104. Most of the work was done as a postdoctoral fellow at the Dept. of Combinatorics
	 and Optimization, Univ. Waterloo.} 
\and
         Chaitanya Swamy\thanks{{\tt cswamy@math.uwaterloo.ca}.
         Dept. of Combinatorics and Optimization, Univ. Waterloo, Waterloo, ON N2L 3G1.
         Supported in part by NSERC grant 327620-09 and an Ontario Early Researcher Award.} 
}
\date{} 

\begin{document}

\maketitle

\begin{abstract}
We introduce a problem that is a common generalization of the 
{\em uncapacitated facility location} (\ufl) and 
{\em minimum latency} (\ml) problems, where facilities not only need to be opened to serve
clients, but also need to be sequentially activated before they can provide service.
This abstracts a setting where inventory demanded by customers needs to be stocked or
replenished at facilities from a depot or warehouse. 
Formally, we are given a set $\F$ of $n$ facilities with facility-opening costs $\{f_i\}$,
a set $\D$ of $m$ clients, and connection costs $\{c_{ij}\}$ specifying the cost of
assigning a client $j$ to a facility $i$, 
a root node $r$ denoting the depot, and a {\em time metric} $d$ on
$\F\cup\{r\}$. Our goal is to open a subset $F$ of facilities,   
find a path $P$ starting at $r$ and spanning $F$ to activate the open facilities, and 
connecting each client $j$ to a facility $\phi(j)\in F$, so as to minimize 
$\sum_{i\in F} f_i + \sum_{j\in \D} (c_{\phi(j),j} + t_j)$,
where $t_j$ is the time taken to reach  $\phi(j)$ along path $P$.
We call this the {\em minimum latency uncapacitated facility location} (\mlufl) problem.  

Our main result is an $O\bigl(\log n\max\{\log n, \log m\}\bigr)$-approximation for
\mlufl. Via a reduction to the group Steiner tree (\gst) problem, we show this result is  
tight in the sense that 
any improvement in the approximation guarantee for \mlufl, implies an improvement in the
(currently known) approximation factor for \gst.
We obtain significantly improved constant approximation guarantees for two natural special
cases of the problem: (a) {\em related \mlufl}, where the connection costs form a metric
that is a scalar multiple of the time metric; 
(b) {\em metric uniform \mlufl}, where we have metric connection costs and 
the time-metric is uniform.  
Our LP-based methods are fairly versatile and are easily adapted with minor changes
to yield approximation guarantees for \mlufl (and \ml) in various more general settings,
such as (i) the setting where the latency-cost of a client is a function (of bounded growth) of
the delay faced by the facility to which it is connected; and 
(ii) the $k$-route version, where we can dispatch $k$ vehicles in parallel to activate the
open facilities. 

Our LP-based understanding of \mlufl also offers some LP-based insights into \ml.  
We obtain two natural LP-relaxations for \ml with constant integrality gap, which we
believe shed new light upon the problem and offer a promising direction for obtaining
improvements for \ml. 
\end{abstract}

\section{Introduction} \label{intro}
Facility location and vehicle routing problems are two broad classes of
combinatorial optimization problems that have been widely studied in the Operations
Research community (see, e.g.,~\cite{MirchandaniF90,TothV02}), and have a wide range of
applications. Both problems can be described in terms of an underlying set of
clients that need to be serviced. In facility location problems, there is a candidate set 
of facilities that provide service, and the goal is to open some facilities and
connect each client to an open facility so as to minimize some combination of the
facility-opening and client-connection costs.
Vehicle routing problems consider the setting where a vehicle
(delivery-man or repairman) provides service, and the goal is to plan a
route that visits (and hence services) the clients as quickly as possible. Two common
objectives considered are: (i) minimize the total length of the vehicle's route, 
giving rise to the traveling salesman problem (\tsp), and 
(ii) (adopting a client-oriented approach) 
minimize the sum of the client delays, giving rise to minimum latency (\ml) problems.    

These two classes of problems have mostly been considered separately. However, various 
logistics problems involve both facility-location and vehicle-routing   
components. For example, consider the 
following oft-cited prototypical example of a facility location problem: 
a company wants to determine where to open its retail outlets so as to serve
its customers effectively. 
Now, inventory at the outlets needs to be replenished or ordered (e.g., from a depot);
naturally, a customer cannot be served by an outlet unless the outlet has the inventory
demanded by it, and  
delays incurred in procuring inventory might adversely impact customers. Hence, it makes
sense for the company to {\em also} keep in mind the latencies faced by the customers 
while making its decisions about where to open outlets,
how to connect customers to open outlets, and in what order to replenish the open outlets,
thereby adding a vehicle-routing component to the problem. 

We propose a mathematical model that is a common generalization of the 
{\em uncapacitated facility location} (\ufl) and {\em minimum latency} (\ml) problems, and  
abstracts a setting (such as above) where facilities 
need to be ``{\em activated}'' before they can provide service. Formally, as in \ufl, we
have a set $\F$ of $n$ facilities, and a set $\D$ of $m$ clients.   
Opening facility $i$ incurs a {\em facility-opening cost} $f_i$, and assigning a client
$j$ to a facility $i$ incurs {\em connection cost} $c_{ij}$. 
Taking a lead from minimum latency problems, we model activation delays as
follows. We have a root (depot) node $r$, and a {\em time metric} $d$ on $\F\cup\{r\}$. 
A feasible solution specifies a subset $F\sse\F$ of facilities to open, a path $P$
starting at $r$ and spanning $F$ along which the open facilities are activated, 
and assigns each client $j$ to an open facility $\phi(j)\in F$. 
The cost of such a solution is 
\vspace{-.75ex}
\begin{equation} 
\sum_{i\in F} f_i + \sum_{j\in\D} \big(c_{\phi(j)j} + t_j\big) \label{obj} 
\vspace{-.75ex} 
\end{equation} 
where $t_j=d_P(r,\phi(j))$ 
is the time taken to reach facility $\phi(j)$ along path $P$. 
We refer to $t_j$ as client $j$'s latency cost. 
The goal is to find a solution with
minimum total cost. 
We call this the {\em minimum-latency uncapacitated facility location} (\mlufl) problem.

Apart from being a natural problem of interest, we find \mlufl appealing since it
generalizes, or is closely-related to, various diverse problems of interest (in addition
to \ufl and \ml);  
our work yields new insights on some of these problems, most notably \ml (see ``Our
results''). 
One such problem, which captures much of the combinatorial core of \mlufl 
is what we call the {\em minimum group latency} (\mgl) problem. Here, we are given an
undirected graph with metric edge weights $\{d_e\}$, 
subsets $\{G_j\}$ of vertices called groups, and a root $r$; the goal is to find a path
starting at $r$ that minimizes the sum of the cover times of the groups, where the cover 
time of $G_j$ is the first time at which some $i\in G_j$ is visited on the path. 
Observe that \mgl can be cast as \mlufl with zero facility costs (where
$\F=\text{node-set}\sm\{r\}$), where for each group $G_j$, we create a client $j$ with 
$c_{ij}=0$ if $i\in G_j$ and $\infty$ otherwise. Note that we may assume that the groups
are disjoint (by creating multiple co-located copies of a node), in which case these
$c_{ij}$s form a metric.
\mgl itself captures various other problems. Clearly, when each $G_j$ is a singleton,
we obtain the minimum latency problem. Also, given a set-cover instance, if we consider a
graph whose nodes are ($r$ and) the sets, create a group $G_j$ for each element $j$
consisting of the sets containing it, and consider the uniform metric, then this \mgl
problem is simply the {\em min-sum set cover} (\mssc) problem~\cite{FeigeLT04}.

\vspace{-2ex}
\paragraph{Our results and techniques.} 
Our main result is an $O\bigl(\log n \max\{\log m,\log n\}\bigr)$-approximation algorithm
for \mlufl (Section~\ref{mlufl-gen}), which for the special case of \mgl, implies an
$O(\log^2 n)$ approximation.  
Complementing this result, we prove (Theorem~\ref{lbthm}) that a $\rho$-approximation
algorithm (even) for \mgl yields an $O(\rho\log m)$-approximation algorithm for the {\em group Steiner tree} (\gst)  problem 
\cite{GKR} 
on $n$ nodes and $m$ groups.  Thus, any improvement in our approximation ratio for \mlufl 
would yield a corresponding
improvement of \gst, whose approximation ratio has remained at $O(\log^2 n\log m)$ for a
decade~\cite{GKR}.  
Moreover, combined with the result of~\cite{HalperinK03} on the inapproximability of \gst,
this shows that \mgl, and hence \mlufl with metric connection costs cannot be approximated
to better than a $\Omega(\log m)$-factor unless \np $\sse$ \ztime$(n^{\polylog(n)})$.  

Given the above hardness result, we investigate certain well-motivated special cases of
\mlufl and obtain significantly improved performance guarantees. 
In Section~\ref{mlufl-related}, we consider the case where
the connection costs form a metric, which is a scalar multiple of the $d$-metric (i.e.,
$d_{uv}=c_{uv}/M$, where $M\geq 1$; the problem is trivial if $M<1$).
For example, in a supply-chain logistics problem, this 
models a natural setting where the connection of clients to facilities, and the activation
of facilities both proceed along the same transportation network. 
We obtain a {\em constant-factor} approximation algorithm for this problem.

In Section~\ref{mlufl-unif}, we consider the {\em uniform \mlufl} problem, which is the
special case where the time-metric is uniform. 
Uniform \mlufl already generalizes \mssc (and also \ufl). 
For uniform \mlufl with metric connection costs (i.e., metric uniform \mlufl), we devise a
$10.78$-approximation algorithm. (Without metricity, the problem becomes set-cover hard,
and we obtain a simple matching $O(\log m)$-approximation.)  
The chief novelty here lies in the technique used to obtain this result.
We give a simple generic reduction (Theorem~\ref{metricredn}) that shows how to reduce the
metric uniform \mlufl problem with facility costs to one {\em without} facility costs,
in conjunction with an algorithm for \ufl. 
This reduction is surprisingly robust and versatile and has other applications. 
For example, the same reduction yields an
$O(1)$-approximation for {\em metric uniform $k$-median} 
(i.e., metric uniform \mlufl where at most $k$ facilities may be opened), 
and the same ideas 
lead  to improved guarantees for the
$k$-median versions of connected facility location~\cite{SwamyK04}, and facility location
with service installation costs~\cite{ShmoysSL04}.

We obtain our approximation bounds 
by rounding the optimal solution to a suitable linear-programming (LP) relaxation
of the problem.  
This is interesting since we are not aware of any previous LP-based methods
to attack \ml (as a whole).
In Section~\ref{lpml}, we leverage this to obtain some interesting
insights about \ml, which we believe cast new light on the problem. 
In particular, we present two LP-relaxations for \ml, and
prove that these have (small) constant integrality gap. 
Our first LP is a specialization of our LP-relaxation for
\mlufl. Interestingly,    
the integrality-gap bound for this LP relies only on the fact that the natural LP relaxation for
\tsp has constant integrality gap (i.e., a $\rho$-integrality gap
for the natural TSP LP relaxation translates to an $O(\rho)$-integrality gap). 
In contrast, the various known algorithms for
\ml~\cite{BC+,CGRT,ArcherLW08} all utilize algorithms for the arguably harder $k$-MST problem or  
its variants. 
Our second LP has exponentially-many variables, one for every path (or tree)
of a given length bound, and  
the separation oracle for the dual problem 
is a rooted path (or tree) orienteering problem: given rewards on the nodes
and metric edge costs, find a (simple) path rooted at $r$ of length at most $B$ that
gathers maximum reward. We prove that even a bicriteria 
approximation for the orienteering problem yields an 
approximation for \ml while losing a constant factor. 
\ignore{This requires circumventing the technical difficulty posed by negative coefficients in the
dual objective, due to
which one cannot make the usual argument that scaling the solution computed via an
approximate separation oracle yields an approximate solution. }
This connection between orienteering and \ml is known~\cite{FHR}. 
But we feel that our alternate proof,  
where the orienteering problem 
appears as the separation oracle required to solve the dual LP, offers a more
illuminating 
explanation of the relation between the approximability of
the two problems. (In fact, the same relationship also holds between \mgl \nolinebreak 
\mbox{and ``group orienteering.'')} 

We believe that the use of LPs opens up \ml to new venues of attack. 
A good LP-relaxation is beneficial because it yields a concrete, tractable lower
bound and handle on the integer optimum, which one can exploit to design algorithms (a
point repeatedly driven home in the field of approximation algorithms).  
Also, LP-based techniques tend to be fairly versatile and can be adapted to handle more
general variants of the problem (more on this below).  
Our LP-rounding algorithms exploit various ideas developed for scheduling and
facility-location problems (e.g., $\al$-points) and polyhedral insights for \tsp, which 
suggests that the wealth of LP-based machinery developed for these problems can be
leveraged to obtain improvements for \ml. 
We suspect that our LP-relaxations are in fact better than what we have accounted for, 
and consider them to be a promising direction for making progress on \ml.  

Section~\ref{extn} showcases the flexibility afforded by our LP-based techniques, by
showing that our algorithms and analyses extend with 
little effort to handle various generalizations of \mlufl (and hence, \ml). 
For example, consider the setting where the latency-cost of a client $j$ is 
$\ld(\text{time taken to reach the facility serving $j$})$, for some non-decreasing 
function $\ld(.)$. 
When $\ld$ is convex and has ``growth'' at most $p$ (i.e.,
$\ld(cx)\leq c^p\ld(x)$), we derive an $O\bigl(\max\{(p\log^2 n)^p,p\log n\log m\}\bigr)$-approximation for \mlufl,  
an $O\bigl(2^{O(p)}\bigr)$-approximation for
related \mlufl and \ml, and an $O(1)$-approximation for metric uniform \mlufl. 
(Concave $\ld$s are even easier to handle.) This in turn leads to approximation guarantees
for the $\Lf_p$-norm generalization of \mlufl, where instead of the sum (i.e.,
$\Lf_1$-norm) of client latencies, the $\Lf_p$-norm of the client latencies appears in the
objective function. The spectrum of $\Lf_p$ norms tradeoff efficiency with fairness,
making the $\Lf_p$-norm problem an appealing problem to consider.  
We obtain an 
$O\bigl(p\log n\max\{\log n,\log m\}\bigr)$-approximation for \mlufl, and
an $O(1)$-approximation for the other special cases. 
Another notable extension is the $k$-route version of the problem, where we may use $k$
paths starting at $r$ to traverse the open facilities. With one simple modification to our
LPs (and algorithms), {\em all} our approximation guarantees translate to this setting. As
a corollary, we obtain a constant-factor approximation for the $\Lf_p$-norm 
version of the {\em $k$-traveling repairmen problem}~\cite{FHR} (the $k$-route version of
\ml). 
 
\vspace{-2ex}
\paragraph{Related work.} 
To the best of our knowledge, \mlufl and \mgl are new problems that have not been
studied previously. 
There is a great deal of literature on facility location and vehicle routing problems
(see, e.g.,~\cite{MirchandaniF90,TothV02}) in general, and \ufl and \ml, in particular,
which are special cases of our problem, and we limit ourselves to a sampling of some of
the relevant results. 
The first constant approximation guarantee for \ufl was obtained by Shmoys, 
Tardos, and Aardal~\cite{ShmoysTA97} via an LP-rounding algorithm, and the current
state-of-the-art is 1.5-approximation algorithm due to Byrka~\cite{Byrka07}. 
The minimum latency (\ml) problem seems to have been first introduced to the computer
science community by Blum et al.~\cite{BC+}, who gave a constant-factor approximation
algorithm for it. Goemans and Kleinberg~\cite{GK} improved the approximation factor to
$10.78$, using a ``tour-concatenation'' lemma, which has formed a component of all 
subsequent algorithms and improvements for \ml. 
The current-best approximation factor for \ml is 3.59 due to Chaudhuri, Godfrey, Rao and 
Talwar~\cite{CGRT}. 
As mentioned earlier, \mgl with a uniform time metric captures the min-sum set cover 
(\mssc) problem.
This problem was introduced by Feige, Lovasz and Tetali~\cite{FeigeLT04}, who gave a
$4$-approximation algorithm for the problem and a matching inapproximability
result. Recently, Azar et al.~\cite{AzarGY09} introduced a generalization of \mssc, for
which Bansal et al.~\cite{BansalGK10} obtained a constant-factor approximation.

If instead of adding up the latency cost of clients, we include the {\em maximum} latency
cost of a client in the objective function of \mlufl, then we obtain the $\min$-$\max$
versions of \mgl and \mlufl, which have been studied previously. 
The $\min$-$\max$ version of \mgl is equivalent to a
``path-variant'' of \gst: we seek a path starting at $r$ of minimum total length that
covers every group. Garg, Konjevod, and Ravi~\cite{GKR}
devised an LP-rounding based $O(\log^2 n\log m)$-approximation for \gst, where $n$ is the
number of nodes and $m$ is the number of groups.  
Charikar et al.~\cite{CCGG} gave a deterministic algorithm with the same guarantee, and 
Halperin and Krauthgamer \cite{HalperinK03} proved an $\Omega(\log^2 m)$-inapproximability
result. The rounding technique of~\cite{GKR} and the deterministic tree-embedding construction
of~\cite{CCGG} (rather its improvement by~\cite{FakcharoenpholRT03}) are key ingredients
of our algorithm for (general) \mlufl.  
The $\min$-$\max$ version of \mlufl can be viewed a path-variant of connected
facility location~\cite{RaviS99,GuptaKKRY01}.  
The connected facility location problem, in its full generality, is essentially equivalent
to \gst~\cite{RaviS99}; however, if the connection costs form a metric, and the
time- and the connection-cost metrics are scalar multiples of each other, then various
constant-factor approximations are known~\cite{GuptaKKRY01,SwamyK04,EisenbrandGRS08}. 

Very recently, we have learnt that, independent of, and concurrent
with, our work, Gupta et al.~\cite{N10}  
also propose the minimum group latency (\mgl) problem (which they arrive at in the
course of solving a different problem), and obtain results similar to ours for \mgl. They
also obtain an $O(\log^2 n)$-approximation for \mgl, and the reduction from \gst to \mgl
with a $\log m$-factor loss (see also~\cite{Nagarajan09}), and relate the approximability
of the \mgl and ``group orienteering'' problems. Their techniques are combinatorial and
not LP-based, and it is not clear how these can be extended to handle 
facility-opening costs.

\section{LP-rounding approximation algorithms for \textsf{MLUFL}} \label{approx}
We can express \mlufl as an integer program and relax the integrality constraints to
obtain a linear program  as follows.  
We may assume that $d_{ii'}$ is integral for all $i,i'\in\F\cup\{r\}$. 
Let $E$ denote the edge-set of the complete graph on $\F\cup\{r\}$ and 
let $d_{max} := \max_{e\in E}d_e$. 
Let $\Time\le\min\{n,m\}d_{\max}$ be a known upper bound on the maximum activation time of
an open facility in an optimal solution. 
For every facility $i$, client $j$, and time $t \le \Time$, 
we have a variable $y_{i,t}$ indicating if facility $i$ is opened at time $t$ or not, and
a variable $x_{ij,t}$ indicating whether client $j$ connects to facility $i$ at time $t$.
Also, for every edge $e\in E$ and time $t$, we introduce a variable $z_{e,t}$ which
denotes if edge $e$ has been traversed {\em by} time $t$. 
Throughout, we use $i$ to index the facilities in $\F$, $j$ to index the clients in $\D$,
$t$ to index the time units in $[\Time]:=\{1,\ldots,\Time\}$, \nolinebreak
\mbox{and $e$ to index the edges in $E$.}
\begin{alignat}{5}
\min & \quad & \sum_{i,t} f_iy_{i,t} & + 
\sum_{j,i,t} \big(c_{ij} + t \big) x_{ij,t}  \tag{P} \label{mlufllp} \\ 
\text{s.t.} && \sum_{i,t} x_{ij,t} & \ge 1 \quad && \frall j; 
\qquad & x_{ij,t} &  \le y_{i,t} \quad && \frall i,j,t  \notag \\
&& \sum_{e} d_ez_{e,t} & \le t \quad && \frall t \label{eq:3} \\
&& \sum_{e\in \delta(S)} z_{e,t} & \ge \sum_{i\in S, t'\le t} x_{ij,t'} 
\quad && \frall t, S\subseteq\F, j \label{eq:4} \\
&& x_{ij,t}, y_{i,t}, z_{e,t} & \ge 0 \quad && \frall i,j,t,e;
\qquad & y_{i,t} & = 0 \quad && \frall i,t\text{ with }d_{ir}>t. \notag 
\end{alignat}
The first two constraints encode that each client is connected to some facility at
some time, and that if a client is connected to a facility $i$ at time $t$, then $i$ 
must be open at time $t$.
Constraint \eqref{eq:3} ensures that by time $t$ no more than $t$ ``distance'' is covered
by the tour on facilities, and \eqref{eq:4} ensures that if a client is connected to $i$
by time $t$, then the tour must have visited $i$ by time $t$. 
We assume for now that $\Time=\poly(m)$, and show later how to remove this assumption
(Lemma~\ref{lpsolve}, Theorem~\ref{scalethm}).
Thus, \eqref{mlufllp} can be solved efficiently since
one can efficiently separate over the constraints \eqref{eq:4}. 
Let $(x,y,z)$ be an optimal solution to \eqref{mlufllp}, and $\OPT$ denote its objective 
value. For a client $j$, define $\bC_j=\sum_{i,t} c_{ij}x_{ij,t}$, 
and $\bL_j=\sum_{i,t}tx_{ij,t}$.  
We devise various approximation algorithms for \mlufl by rounding $(x,y,z)$ to
an integer solution. 

In Section~\ref{mlufl-gen}, we give a polylogarithmic
approximation algorithm for (general) \mlufl (where the $c_{ij}$s need not
even form a metric). Complementing this result, we prove (Theorem~\ref{lbthm}) that a
$\rho$-approximation algorithm (not necessarily LP-based) for \mlufl yields an 
$O(\rho\log m)$-approximation algorithm for the \gst problem on $n$ nodes and $m$ groups.  
In Sections~\ref{mlufl-related} and~\ref{mlufl-unif}, we obtain significantly-improved
approximation guarantees for various well-motivated special cases of \mlufl.    
Section~\ref{mlufl-related} obtains a {\em constant-factor} approximation algorithm in the
natural setting where the connection costs form a metric that is a scalar multiple of the 
time-metric. 
Section~\ref{mlufl-unif} considers the setting where the time-metric is the uniform
metric. 
Our main result here is a constant-factor approximation for metric connection costs, which
is obtained via a rather versatile reduction of this uniform \mlufl problem to \ufl and
uniform \mlufl with zero-facility costs.

\subsection{An \boldmath $O\bigl(\log n\cdot\max\{\log n,\log m\}\bigr)$-approximation
algorithm} \label{mlufl-gen}
We first give an overview. 
Let $N_j = \{i\in\F: c_{ij} \le 4\bC_j\}$ be the set of facilities ``close" to $j$, and define
$\tau_j$ as the earliest time $t$ such that 
$\sum_{i\in N_j,t'\leq t}x_{ij,t'}\geq\frac{2}{3}$. 
By Markov's inequality, we have $\sum_{i\in N_j}\sum_t x_{ij,t}\geq\frac{3}{4}$ and
$\tau_j\leq 12\bL_j$.
It is easiest to describe the algorithm assuming first that the time-metric
$d$ is a tree metric. 
Our algorithm runs in phases, with phase $\ell$ corresponding to time $t_\ell=2^\ell$.  
In each phase, we compute a random subtree 
rooted at $r$ of ``low'' cost 
such that for every client $j$ with $\tau_j\leq t_\ell$, with constant probability, 
this tree contains a facility in $N_j$. 
To compute this tree, we utilize the rounding procedure of Garg-Konjevod-Ravi (GKR) for 
the {\em group Steiner tree} (\gst) problem~\cite{GKR} (see Lemma~\ref{lem:gkr}), by
creating a group for each client $j$ with $\tau_j\leq t_\ell$ comprising of, roughly
speaking, the facilities in $N_j$. 
We open all the facilities included in the subtree, 
and obtain a tour via the standard trick of doubling all edges and performing  
an Eulerian tour with possible shortcutting. The overall tour is a concatenation of all
the tours obtained in the various phases. 
For each client $j$, we consider the first tree that contains a facility 
from $N_j$ (which must therefore be open), and connect $j$ to such a facility. 

Given the result for tree metrics, an oft-used idea to handle the case when $d$
is not a tree metric is to approximate it by a distribution of tree metrics with 
$O(\log n)$ distortion~\cite{FakcharoenpholRT03}. Our use of this idea is however
more subtle than the typical applications of probabilistic tree embeddings. Instead of 
moving to a distribution over tree metrics up front,
in each phase $\ell$, we use the results of~\cite{CCGG,FakcharoenpholRT03} to
{\em deterministically} obtain a tree $\T_\ell$ with edge weights $\{d_{\T_\ell}(e)\}$,
such that the resulting tree metric dominates $d$ and 
$\sum_{e=(i,i')} d_{\T_\ell}(i,i')z_{e,t_\ell}=O(\log n)\sum_{e} d_ez_{e,t_\ell}$.
As we show in Section~\ref{extn}, this
deterministic choice allows to extend our algorithm and analysis effortlessly to the
setting where the latency-cost in the objective function is measured by a more general
function (e.g., the $\Lf_p$-norm) of the client-latencies.  
The algorithm is described in detail as Algorithm~\ref{genalg}, and utilizes the following
results.  
Let $\tau_{\max}=\max_j\tau_j$. 

\begin{theorem}[\cite{CCGG,FakcharoenpholRT03}] \label{thm:frt}
Given any edge weights $\{\z_e\}_{e\in E}$, one can deterministically construct a
weighted tree $\T$ having leaf-set $\F\cup\{r\}$, leading to a tree metric, $d_{\T}(\cdot)$,
such that, for any $i,i'\in\F\cup\{r\}$, we have:  

\noindent
(i) $d_\T(i,i') \ge d_{i,i'}$, and 
(ii) $\sum_{e=(i,i')\in E}d_\T(i,i')\z_{i,i'} = O(\log n)\sum_{e}d_e\z_e$. 
\end{theorem}

\begin{theorem}[\cite{GKR}] \label{thm:gkr}
Consider a tree $\T$ rooted at $r$ with $n$ leaves, subsets $G_1,\ldots,G_p$ of leaves,
and fractional values $\z_e$ on the edges of $\T$ satisfying
$\z(\delta(S)) \ge \nu_j$ for every group $G_j$ and node-set $S$ such that $G_j\sse S$, 
where $\nu_j\in\bigl[\frac{1}{2},1\bigr]$. There exists a randomized polytime algorithm,
henceforth called the GKR algorithm, that returns a rooted subtree $T'' \subseteq \T$
such that (i) $\Pr[e\in T'']\leq\z_e$ for every edge $e\in\T$; and
(ii) $\Pr[T''\cap G_j=\es]\leq\exp\bigl(-\frac{\nu_j}{64\log_2 n}\bigr)$ 
for every group $G_j$.
\end{theorem}

\vspace{-5pt}
\begin{figure}[ht!]
\small
\hrule \vspace{5pt}
\begin{algorithm} \label{genalg}
Given: a fractional solution $(x,y,z)$ to \eqref{mlufllp} (with $\bC_j, \bL_j, N_j$, and   
$\tau_j$ defined as above for each client $j$).
\end{algorithm}
\vspace{-14pt}
\begin{labellist}{A}
\item In each phase $\ell=0,1,\ldots,\Nc:=\ceil{\log_2 (2\tau_{\max})+4\log_2 m}$, we
do the following. Let $t_\ell=\min\{2^\ell,\Time\}$. 
\vspace{-1ex}
\begin{list}{A\arabic{enumi}.\arabic{enumii}.}{\usecounter{enumii} \addtolength{\leftmargin}{-3ex}}
\item Use Theorem~\ref{thm:frt} with edge weights $\{z_{e,t_\ell}\}$ to obtain a tree 
$\T_\ell=\bigl(V(\T_\ell),E(\T_\ell)\bigr)$. 
Extend $\T_\ell$ to a tree $\T'_\ell$ by 
adding a dummy leaf edge $(i,v_i)$ of cost $f_i$ to $\T_\ell$ for each facility $i$. 
Let $E'=\{(i,v_i): i\in\F\}$. 

\item Map the LP-assignment $\{z_{e,t_\ell}\}_{e\in E}$ 
to an assignment $\z$ on the edges of $\T'_\ell$ by setting 
$\z_e=\sum_{e\text{ lies on the unique $i$-$i'$ path in }\T_\ell}z_{ii',t_\ell}$ for 
all $e\in E(\T_\ell)$, and $\z_e = \sum_{t\leq t_\ell} y_{i,t}$ for all $e=(i,v_i)\in E'$.  
Note that $\sum_{e\in E(\T_\ell)}d_{\T_\ell}(e)\z_e=
\sum_{e=(i,i')}d_{\T_\ell}(i,i')z_{e,t_\ell}=O(\log n)\sum_e d_ez_{e,t_{\ell}}
=O(\log n)t_\ell$.

\item Define $D_\ell=\{j: \tau_j\leq t_\ell\}$. 
For each client $j\in D_\ell$, we define the group $N'_j=\{v_i: i\in N_j\}$. 
We now compute a subtree $T'_\ell$ of $\T'_\ell$ as follows. 
We obtain $N:=\log_2 m$ subtrees $T''_1,\ldots,T''_N$. Each tree $T''_r$ is obtained by 
executing the GKR algorithm 
$192\log_2 n$ times on the tree $\T'_\ell$ with groups $\{N'_j\}_{j\in D_\ell}$, and taking
the union of all the subtrees returned. Note that we may assume that $i\in T''_r$ iff
$(i,v_i)\in T''_r$.  
Set $T'_\ell$ to be any tree in $\{T''_1,\ldots,T''_N\}$ satisfying 
(i) $\sum_{(i,v_i)\in E(T'_\ell)}f_i\leq 40\cdot 192\log_2 n\sum_{(i,v_i)\in E'}f_i\z_{i,v_i}$ 
and 
(ii) $\sum_{e\in E(T'_\ell)\sm E'}d_{\T_\ell}(e)\leq 
40\cdot 192\log_2 n\sum_{e\in E(\T_\ell)}d_{\T_\ell}(e)\z_e$; if no such tree exists, the
algorithm fails.

\item Now remove all the dummy edges from $T'_\ell$, 
open all the facilities in the resulting tree, and convert the resulting tree into a tour 
$\Tour_\ell$ traversing all the opened facilies.  
For every unconnected client $j$, we connect $j$ to a facility in $N_j$ if some such
facility is open (and hence part of $\Tour_\ell$). 
\end{list}

\vspace{-0.5ex}
\item Return the concatenation of the tours $\Tour_\ell$ for $\ell=0,1,\ldots,\Nc$
shortcutting whenever possible. This induces an ordering of the open facilities. If some
client is left unconnected, we say that the algorithm has failed.
\end{labellist}
\hrule \vspace{-10pt}
\end{figure}

\vspace{-2ex}
\paragraph{Analysis.}
The algorithm may fail in steps A.1.3 and A2. Lemmas~\ref{lem:gkr} and~\ref{lbound} bound
the failure probability in each case by $1/\poly(m)$.
To bound the expected cost conditioned on success, it suffices to bound the expectation of
the random variable that equals the cost incurred if the algorithm succeeds, and is 0 otherwise.     

Since each client $j$ is connected to a facility in $N_j$, the total connection cost is
at most $4\sum_j\bC_j$. 
To bound the remaining components of the cost, we first show that 
in any phase $\ell$, the $\z$-assignment defined above ``covers'' each
group in $\{N'_j\}_{j\in D_\ell}$ to an extent of at least $\frac{2}{3}$
(Claim~\ref{zcover}). Next, we show in Lemma~\ref{lem:gkr} 
that for every client $j\in D_\ell$, the probability that a facility in $N_j$ is included
in the tree $T'_\ell$, and hence opened in phase $\ell$, is at least 
$\frac{5}{6}\cdot\frac{2}{3}=\frac{5}{9}$. 
The facility-cost incurred in a phase is $O(\log n)\sum_{i,t}f_iy_{i,t}$, and   
since $\tau_{\max}\leq\Time=\poly(m)$, the number of phases is $O(\log m)$, 
so this bounds the facility-opening cost incurred. 
Also, since the probability that $j$ is not connected (to a facility in $N_j$) in phase
$\ell$ decreases geometrically (at a rate less than $1/2$) with $\ell$ when $t_\ell\geq\tau_j$, one can argue that
(a) with very high probability (i.e., $1-1/\poly(m)$), each client $j$ is connected to some
facility in $N_j$, and (b)  
the expected latency-cost of $j$ is at most 
$O(\log n)\sum_{e\in E(\T_\ell)}d_{\T_\ell}(e)\z_e=O(\log^2 n)\tau_j$.

\begin{claim} \label{zcover}
Consider any phase $\ell$.
For any subset $S$ of nodes of the corresponding tree $\T'_\ell$ with $r\notin S$, and any   
$N'_j \subseteq S$ where $j\in D_\ell$, we have 
$\z(\delta(S)) \ge\sum_{i\in N_j,t\leq t_\ell}x_{ij,t}\geq 2/3$ (where $\dt(S)$
denotes $\dt_{\T'_\ell}(S)$).
\end{claim}
\begin{proof}
Let $R=S\cap V(\T_\ell)$, and $Y\sse\F$ be the set of leaves in $R$. 
Then $\dt(S)=\dt_{\T_\ell}(R)\cup\bigl(\dt(S)\cap E'\bigr)$.
Let $\dt_E(Y)=\{(i,i')\in E: |\{i,i'\}|\cap Y|=1\}$.
Observe that $\z(\dt_{\T_\ell}(R))\geq\sum_{e\in\dt_E(Y)}z_{e,t_{\ell}}$. This is simply
because if we send $z_{ii',t_\ell}$ flow along the unique $(i,i')$ path in $\T_\ell$ for
every $(i,i')\in\dt_E(Y)$, then we obtain a flow between $Y$ and $\F\cup\{r\}\sm Y$
respecting the capacities $\{\z_e\}_{e\in E(\T_\ell)}$, and of value equal to the RHS
above. Thus, the inequality follows because the capacity of any cut containing $Y$ must be 
at least the value of the flow. Since $(x,y,z)$ satisfies \eqref{eq:4}, we further have
that $\sum_{e\in\dt_E(Y)}z_{e,t_{\ell}}\geq\sum_{i\in Y,t\leq t_{\ell}}x_{ij,t}$. 
So
$$
\z(\delta(S))\geq\sum_{e\in\dt_{\T_\ell}(R)}\z_{e,t_\ell}+\hspace*{-2ex}
\sum_{(i,v_i)\in\dt(S):i\notin Y}\hspace*{-1ex}\bigl(\sum_{t\leq t_\ell} y_{i,t}\bigr) 
\geq\sum_{i\in Y,t\leq t_\ell}x_{ij,t}+\hspace{-2ex}
\sum_{i\in N_j\setminus Y,t\le t_\ell}\hspace{-2ex}x_{ij,t}
\geq\sum_{i\in N_j,t\leq t_\ell}x_{ij,t}\geq\frac{2}{3}. \vspace{-3ex}
$$
\end{proof}

\begin{lemma}\label{lem:gkr}
In any phase $\ell$, with probability $1-1/\poly(m)$, we obtain the desired tree $T'_\ell$
in step A1.3. Moreover, $\Pr[T'_\ell\cap N'_j\neq\es]\geq 5/9$ for all $j\in D_\ell$.
\end{lemma}
\begin{proof}
Consider any tree $T''_r$ obtained by executing the GKR algorithm $192\log_2 n$ times and
taking the union of the resulting subtrees. For brevity, we denote 
$\sum_{(i,v_i)\in E(T''_r)}f_i$ by $F(T''_r)$, and 
$\sum_{e\in E(T''_r)\sm E'}d_{\T_\ell}(e)$ by $d_\ell(T''_r)$.
 
For $j\in D_\ell$, let ${\mathbf E}^r_j$
denote the event that $T''_r \cap N'_j$ is non-empty. 
By Theorem~\ref{thm:gkr}, we have 
$\Pr[\bdE^r_j]\geq 1-\exp\bigl(-\frac{\nu_j}{64\log_2 n}\cdot 192\log n\bigr)\geq
1-e^{-3\nu_j}\geq (1-e^{-3})\nu_j\geq 11/18$.
We also have that $\E{F(T''_r)}\leq 192\log_2 n\sum_{(i,v_i)\in E'}f_i\z_{i,v_i}$ and
$\E{d_\ell(T''_r)}\leq 192\log_2 n\sum_{e\in E(\T_\ell)}d_{\T_\ell}(e)\z_e$.
Let ${\mathbf F}^r$ and ${\mathbf D}^r$ denote respectively the events that 
$F(T''_r) \le 40\cdot 192\log_2 n \sum_{(i,v_i)\in E'} f_i\z_{i,v_i}$, and
$d_\ell(T''_r)\leq 40\cdot 192\log_2 n\sum_{e\in E(\T_\ell)}d_{\T_\ell}(e)\z_e$.
By Markov's inequality each event happens with probability at least $39/40$. Thus, for any
$j\in D_\ell$, we get that 
\begin{equation}\label{eq:10}
\Pr[{\mathbf E}^r_j|({\mathbf F}^r\wedge {\mathbf D}^r)] \ge 
\Pr[{\mathbf E}^r_j \wedge {\mathbf F}^r\wedge {\mathbf D}^r] \ge 
1 - (7/18 + 1/40 + 1/40) > 5/9.
\end{equation}

Now the probability that $(\bdF^r\wedge\bdD^r)^c$ happens for all $r=1,\ldots,N$
is at most $(2/40)^N\leq 1/m^4$. Hence, with high probability, there is
some tree $T'_\ell:=T_r$ such that both $\bdF^r$ and $\bdD^r$ hold, and \eqref{eq:10}
shows that $\Pr[T'_\ell\cap N'_j\neq\es]\geq 5/9$ for all $j\in D_\ell$.
\end{proof}

\begin{lemma} \label{lbound}
The probability that a client $j$ is not connected by the algorithm is at most $1/m^4$.
Let $L_j$ be the random variable equal to $j$'s latency-cost if the algorithm
succeeds and $0$ otherwise. 
Then $\E{L_j}=O(\log^2 n)t_{\ell_j}$, where
$\ell_j$ ($=\ceil{\log_2 \tau_j}$) is the smallest $\ell$ such that $t_\ell\geq\tau_j$.  
\end{lemma}

\begin{proof}
Let $P_j$ be the random variable denoting the phase in which $j$ gets connected; let
$P_j:=\Nc+1$ if $j$ remains unconnected.
We have $\Pr[P_j\geq\ell]\leq\bigl(\frac{4}{9}\bigr)^{(\ell-\ell_j)}$ for $\ell\geq\ell_j$ 
The algorithm proceeds for at least $4\log_2 m$ phases after phase $\ell_j$,
so $\Pr[\text{$j$ is not connected after $\Nc$ phases}]\leq 1/m^4$. 
Now,
\begin{eqnarray*} 
\hspace*{-2ex}
L_j \leq \sum_{\ell\leq P_j}d(\Tour_\ell) \leq 
2\sum_{\ell\leq P_j}\sum_{e\in E(T'_\ell)\sm E'}d_{\T_\ell}(e)
\ =\ O(\log n)\sum_{\ell\leq P_j}\sum_{e\in E(\T_\ell)}d_{\T_\ell}(e)\z_e \ =\ O(\log^2 n)\sum_{\ell\leq P_j} t_\ell  \\
\text{so} \ \ \E{L_j} = O(\log^2 n)\sum_{\ell=0}^{\Nc}\Pr[P_j\geq\ell]\cdot t_\ell
\leq O(\log^2 n)
\biggl[\sum_{\ell=0}^{\ell_j} t_\ell + 
\sum_{\ell>\ell_j}t_\ell\cdot
\left(\frac{4}{9}\right)^{(\ell-\ell_j)}\biggr]
= O(\log^2 n)t_{\ell_j}. 
\end{eqnarray*}

\vspace{-6ex}
\end{proof}

\begin{theorem} \label{genthm}
Algorithm~\ref{genalg} succeeds with probability $1-1/\poly(m)$, and returns a
solution of expected cost $O\bigl(\log n\cdot\max\{\log n,\log m\}\bigr)\cdot\OPT$.  
\end{theorem}

\begin{proof} 
Lemmas~\ref{lem:gkr} and~\ref{lbound} show that the failure probability is $1/\poly(m)$.
Let $Y$ denote the cost incurred if the algorithm succeeds, and 0 otherwise. 
Since $t_{\ell_j}\leq 2\tau_j=O(\bL_j)$ for each $j$, we have 
$\E{Y}=O\bigl(\log n\cdot\max\{\log n,\log m\}\bigr)\cdot\OPT$ by Lemma~\ref{lbound} and
the preceding arguments. 
\end{proof}

\paragraph{Removing the assumption \boldmath $\Time=\poly(m)$.}
We first argue that although
\eqref{mlufllp} has a pseudopolynomial number of variables, one can compute a near-optimal
solution to it in polynomial time (Lemma~\ref{lpsolve}) by considering only (integer)
time-values that are powers of $(1+\e)$ (rouhgly speaking).  
Given $\e>0$, define $\Time_r=\ceil{(1+\e)^r}$, and let 
$\TS:=\{\Time_0,\Time_1,\ldots,\Time_k\}$ where $k$ is the smallest integer such that
$\Time_k\geq \min\{n,m\}d_{\max}$. Define $\Time_{-1}=0$.
Let \lp{\TS} denote \eqref{mlufllp} when $t$ ranges over $\TS$. 
\begin{lemma} \label{lpsolve}
For any $\e>0$, we can obtain a solution to \eqref{mlufllp} of cost
at most $(1+\e)\OPT$ in time $\poly(\text{input size},1/\e)$.
\end{lemma}
\begin{proof}
We prove that the optimal value of \lp{\TS} is at most $(1+\e)\OPT$.
Since the size of \lp{\TS} is $\poly(\text{input size},1/\e)$ this proves the lemma.

We transform $(x,y,z)$, an optimal solution to \eqref{mlufllp} to a feasible solution
$(x',y',z')$ to \lp{\TS} of cost at most $(1+\e)\OPT$. 
(In fact, the facility-opening and connection-costs remain unchanged, and the latency-cost
blows up by a $(1+\e)$-factor.) $z'$ is simply a restriction of $z$
to the times in $\TS$, that is, $z'_{e,t}=z_{e,t}$ for each $e, t\in\TS$. 
Set $x'_{ij,1}=x_{ij,1},\ y'_{i,1}=y_{i,1}$ for all $i$ and $j$. 
For each $\ell=1,\ldots,k$, facility $i$, client $j$, 
we set $x'_{ij,\Time_\ell}=\sum_{t=\Time_{\ell-1}+1}^{\Time_\ell}x_{ij,t}$ and
$y'_{i,\Time_\ell}=\sum_{t=\Time_{\ell-1}+1}^{\Time_\ell}y_{i,t}$.
It is clear that $\sum_{t\in\TS}x'_{ij,t}=\sum_{t} x_{ij,t}$ and
$\sum_{t\in\TS}y'_{i,t}=\sum_t y_{i,t}$ for all $i$ and $j$, and moreover for any
$\Time_\ell\in\TS$, we have
$\sum_{t\in\TS:t\leq\Time_\ell}x'_{ij,t}=\sum_{t\leq\Time_\ell}x_{ij,t}$. It follows that
$(x',y',z')$ is a feasible solution to \lp{\TS} and
$\sum_{i,t\in\TS}f_iy'_{i,t}=\sum_{i,t}f_iy_{i,t},\
\sum_{j,i,t\in\TS}c_{ij}x'_{ij,t}=\sum_{j,i,t}c_{ij}x_{ij,t}$.
To bound the latency cost, note that for any $t>\Time_{\ell-1}$, we have
$\Time_\ell\leq(1+\e)t$, so for any client $j$ and facility $i$, 
$\sum_{t\in\TS}tx'_{ij,t} 
\leq x_{ij,1}+(1+\e)\sum_{t>1}tx_{ij,t}\leq(1+\e)\sum_t tx_{ij,t}$.  
Thus, $\sum_{j,i,t\in\TS}tx'_{ij,t}\leq(1+\e)\sum_{j,i,t}tx_{ij,t}$.
\end{proof}

Let $(x',y',z')$ denote an optimal solution to \lp{\TS}.
The only changes to Algorithm~\ref{genalg} are in the definition of the time $t_\ell$ and
the number of phases $\Nc$. 
(Of course we now work with $(x',y',z')$, and $\bC_j,\ \bL_j$ are defined in terms of
$(x',y',z')$ now.) 
The idea is to define $t_\ell$ so that one can ``reach'' the $\tau_j$ of every
client $j$ in $O(\log m)$ phases; thus, one can terminate in $O(\log m)$ phases and
thereby obtain the same approximation on the facility-opening cost.
Let $\bbL=(\sum_j\bL_j)/m=(\sum_{j,i,t\in\TS}tx'_{ij,t})/m$. 
For $x\leq\Time_k$, define $\TS(x)$ to be the earliest time in $\TS$ that is at least $x$;
if $x\geq\Time_k$, define $\TS(x):=\Time_k$. 
Note that $\TS(x)\leq(1+\e)x$ for all $x\geq 0$.
We now define $t_\ell=\TS(\bbL\cdot 2^\ell)$, and set the number of phases to 
$\Nc:=\ceil{\log_2(2\tau_{\max}/\bbL)+4\log_2 m}$.
Note that since $\tau_j = O(\bL_j)$, we have $\Nc = O(\log m)$.
\begin{theorem} \label{scalethm}
For any $\e>0$, Algorithm~\ref{genalg} with the above modifications succeeds
with high probability and returns a solution of expected cost 
$O\bigl(\log n\max\{\log n,\log m\}\bigr)(1+\e)\OPT$.
\end{theorem}
\begin{proof}
The analysis of the facility-opening and connection-cost is exactly as before (since
$\Nc=O(\log m)$). Define $\ell_j$ as the smallest $\ell$ such that $t_{\ell}\geq\tau_j$.   
Note that $\ell_j\leq\floor{\log_2(2\tau_j/\bbL)}$ (this holds even when $\tau_j\leq\bbL$).
Hence, the probability that $j$ is not connected after $\Nc$ phases is at most $1/m^4$. So  
as before, the failure probability is at most $1/\poly(m)$.
We have $\sum_{\ell\leq\ell_j}t_\ell\leq 
(1+\e)\bbL\sum_{\ell\leq\ell_j}2^\ell\leq 2(1+\e)t_{\ell_j}$, and 
$\sum_{\ell>\ell_j}t_\ell\bigl(\frac{4}{9}\bigr)^{(\ell-\ell_j)}=O(t_{\ell_j})$.
Thus, the inequalities involving $L_j$ and $\E{L_j}$ in Lemma~\ref{lbound} are still
valid, and we obtain the same bound on $\E{L_j}$ as in Lemma~\ref{lbound}.
Note that 
$t_{\ell_j}\leq 2(1+\e)\tau_j=O(\bL_j)$ when $\tau_j\geq\bbL$. Thus, 
$\sum_j\E{L_j}\leq O(\log^2 n)\bigl[m\cdot\bbL+\sum_{j:\tau_j>t_0}O(\bL_j)\bigr]
=O(\log^2 n)\bL$.   
\end{proof}

\paragraph{Inappproximability of \mlufl.} \label{mlufl-lb} 
We argue that any improvement in the guarantee obtained in Theorem~\ref{genthm} would
yield an improvement in the approximation factor for \gst. 
We reduce \gst to \mgl, the special case of \mlufl mentioned in Section~\ref{intro}, where
we have groups $G_j\sse\F$ and the goal is to order the the facilities so as to minimize
the sum of the covering times of the groups. (Note that Theorem~\ref{genthm} implies an
$O(\log^2 n)$-approximation  for \mgl.)
Recall that we may assume that the groups in \mgl are disjoint, in which case
the connection costs form a metric. 

\begin{theorem} \label{lbthm}
Given a $\rho_{n,m}$-approximation algorithm for \mgl with (at most) $n$ nodes and $m$ 
groups, we can obtain an $O(\rho_{n,m}\log m)$-approximation algorithm for \gst with $n$ 
nodes and $m$ groups. Thus, the polylogarithmic inapproximability of
\gst~\cite{HalperinK03} implies that \mgl, and hence \mlufl even with metric connection
costs, cannot be approximated to a factor better than $\Omega(\log m)$, even when the time-metric
arises from a hierarchically well-separated tree, unless \np $\sse$ \ztime$(n^{\polylog(n)})$. 
\end{theorem}   

Our proof of the above theorem is LP-based and is deferred to Appendix~\ref{append-gen}.
Gupta et al.~\cite{N10} independently arrived at the above theorem via a
combinatorial proof.

\subsection{\mlufl with related metrics} \label{mlufl-related}
Here, we consider the \mlufl problem when 
the facilities, clients, and the root $r$ are located in a common metric
space that defines the connection-cost metric (on $\F\cup\D\cup\{r\}$), and we have
$d_{uv}=c_{uv}/M$ for all $u,v\in\F\cup\D\cup\{r\}$. 
We call this problem, {\em related \mlufl}, and design an \nolinebreak
\mbox{$O(1)$-approximation algorithm for it.}  

The algorithm follows a similar outline as Algorithm~\ref{genalg}.
As before, 
we build the tour on the open facilities by 
concatenating tours obtained by ``Eulerifying'' trees rooted at $r$ of geometrically
increasing length. At a high level, the improvement in the approximation 
arises because one can now obtain these
trees without resorting to Theorems~\ref{thm:frt} and \ref{thm:gkr}  
and losing $O(\log n)$-factors in process. Instead, since the $d$- and $c$- 
metrics are related, we can obtain a group Steiner tree on the relevant groups
by using a Steiner tree algorithm  (in a manner similar to the LP-rounding algorithms
in~\cite{RaviS99,GuptaKKRY01}). 
We now define $N_j=\{i: \sum_{t}x_{ij,t}>0,\ c_{ij}\leq 3\bC_j\}$, and $\tau_j=6\bL_j$.
Ideally, in each phase $\ell$, we want to connect the $N_j$ groups for all $j$ such that
$\tau_j \le t_\ell := 2^\ell$. But to obtain a low-cost solution, we 
do a facility-location-style clustering of the set of clients with 
$\tau_j\leq t_\ell$ and build a tree $\T_\ell$ that connects the $N_j$s of
only the cluster centers: we contract these $N_j$s and build an MST (in the $d$-metric) on
them, and then connect each $N_j$ internally (in the $d$-metric) using intracluster edges
incident on $j$.   
{\em Here we crucially exploit the fact that the $d$- and $c$-metrics are related}. 

Deciding which facilities to open is tricky 
because groups $N_j$ and $N_k$ created in different phases could overlap, and we could
have $\bC_j\ll\bC_k$ but $\tau_j\gg\tau_k$; so if we open $i\in N_j$ and use this to also 
serve $k$, then we must connect $i$ to some $\T_\ell$ (without increasing its $d$-cost by
much) where $t_\ell=O(\tau_k)$. 
We consider the collection $\C$ of cluster centers
created in all the phases, and pick a maximal subset $\C' \sse\C$ that yields disjoint
clusters by greedily considering clusters in increasing $\bC_j$ order. 
We open the cheapest facility in $N_j$ for all $j\in \C'$, 
and attach it to the tree $\T_\ell$, where $\ell$ is the {\em earliest} phase
such that there is some cluster $N_k$ created in that phase that was removed (from $\C$)
when $N_j$ was included in $\C'$ (because $N_k\cap N_j\neq\es$). 
Since $d$ and $c$ are related, one can bound the resulting increase in the $d$-cost of
$\T_\ell$. Finally, we convert these augmented $\T_\ell$-trees to tours and concatenate
these tours. 

We now describe the algorithm in detail. We have not sought to optimize the approximation
ratio. When we refer to an edge or a node  
below, we mean an edge or node of the complete graph on $\F\cup\D\cup\{r\}$. 
Recall that $N_j=\{i: \sum_{t}x_{ij,t}>0,\ c_{ij}\leq 3\bC_j\}$, and $\tau_j=6\bL_j$.
So $\sum_{i\in N_j}\sum_{t}x_{ij,t}\geq\frac{2}{3}$, and 
$\sum_{i\in N_j,t\leq\tau_j}x_{ij,t}\geq\frac{1}{2}$. 
\begin{list}{R\arabic{enumi}.}{\usecounter{enumi} \addtolength{\leftmargin}{-1ex}}
\item For each time $t_\ell=2^\ell$, where $\ell=0,1,\ldots,\lceil\log\Time\rceil$, we do 
the following. Define 
$D_\ell=\{j: \tau_j\leq t_\ell\}\setminus \bigl(\bigcup_{0\leq \ell'<\ell}\C_{\ell'}\bigr)$ 
(where the union of an empty collection is $\es$).
\vspace{-0.5ex}
\begin{list}{$\bullet$}{\usecounter{enumii} \addtolength{\leftmargin}{-5ex}}
\item {\bf (Clustering)\ } We cluster the facilities in $\bigcup_{j\in D_\ell}N_j$ as
follows. We pick $j\in D_\ell$ with smallest $\bC_j$ value and form a cluster around $j$
consisting of the facilities in $N_j$. For every client $k\in D_\ell$ (including $j$) such
that $c_{jk}\leq 30\bC_k$ (note that $\bC_k\geq\bC_j$), we remove $k$ from $D_\ell$, set
$\sg(k)=j$, and recurse on the remaining clients in $D_\ell$ until no client is left in
$D_\ell$. Let $\C_\ell$ denote the set of cluster centers (i.e., 
$\{j\in D_\ell: \sg(j)=j\}$).  Note that for two clients $j$ and $j'$ in $\C_\ell$, $N_j\cap N_{j'}$ 
is $\es$.

\item {\bf (Building a group Steiner tree $\T_\ell$ on $\{N_j\}_{j\in\C_\ell}$)\ } We
contract the clusters $N_j$ for $j\in\C_\ell$ into supernodes, and build a minimum
spanning tree (MST) $\T''_\ell$ connecting $r$ and these supernodes. 
Next, we uncontract the supernodes, and for each $j\in\C_\ell$, we add edges joining $j$
to every facility $i\in N_j$ that has an edge incident to it in $\T''_\ell$. This yields
the tree $\T_\ell$.
\end{list}

\item {\bf (Opening facilities)\ } Let $\C=\bigcup_\ell \C_{\ell}$. (Note that a client
appears in at most one of the $\C_\ell$ sets.) We cannot open a facility in every cluster
centered around a client in $\C$, since for $j\in \C_\ell$ and $k\in \C_{\ell'}$, $N_j$ and 
$N_{k}$ need not be disjoint.
So we select a subset $\C'\sse\C$ such that for any two $j,k$ in $\C'$, the sets
$N_j$ and $N_{k}$ are disjoint. This is done as follows. We initialize
$\C'\assign\es$. Pick the client $j\in\C$ with smallest $\bC_j$ value and add it to
$\C'$. We delete from $\C$ every client $k\in\C$ (including $j$) such that 
$N_k\cap N_j\neq\es$, setting $\nbr(k)=j$, and recurse on the remaining set of clients
until no client is left in $\C$.  

Consider each $j\in\C'$. We open the facility $i\in N_j$ with smallest $f_i$. Let $\ell$
be the smallest index such that there is some $k\in\C_\ell$ with $\nbr(k)=j$. We connect
$i$ to $\T_\ell$ by adding the {\em facility edge} $(i,k)$ to it. Let $\T'_\ell$ denote
$\T_\ell$ augmented by all such facility edges.

\item We obtain a tour connecting all the open facilities, by converting each tree
$\T'_\ell$ into a tour, and concatenating the tours for $\ell=0,\ldots,\lceil\Time\rceil$
(in that order).

\item For every client $j\in\C$, we assign $j$ to the facility opened from
$N_{\nbr(j)}$. For every client $j\notin\C$, we assign $j$ to the same facility as
$\sg(j)$. 
\end{list}

\begin{theorem} \label{relmetthm}
For related \mlufl, 
one can round $(x,y,z)$ to get a solution with facility-opening cost at most
$\frac{3}{2}\sum_{i,t}f_iy_{i,t}$, where each client $j$ incurs connection-cost at most
$39\bC_j$ and latency-cost at most $64\tau_j=384\bL_j$. 
Thus, we obtain an $O(1)$-approximation algorithm for related \mlufl. 
\end{theorem}
\begin{proof}
The clusters $N_j$ for clients $j\in\C'$ are disjoint; each such cluster has facility
weight $\sum_{i\in N_j}\sum_t y_{i,t}\geq\frac{2}{3}$, and we open the cheapest facility
in the cluster.

Consider a client $j$, and let it be assigned to facility $i$. If $j\in\C'$, then 
$i\in N_j$, and we have $c_{ij}\leq 3\bC_j$. If $j\in\C\sm\C'$ with $\nbr(j)=k$, then
$i\in N_k$ and there is some facility $i'\in N_j\cap N_k$. So 
$c_{ij}\leq c_{ik}+c_{i'k}+c_{i'j}\leq 2\cdot 3\bC_k+3\bC_j\leq 9\bC_j$.
Finally, if $j\notin\C$ and $j'=\sg(j)$, then $\sg(j')$ is also assigned to $i$, so
$c_{ij}\leq c_{ij'}+c_{jj'}\leq 9\bC_{j'}+30\bC_j\leq 39\bC_j$. 

We next bound $d(\T_\ell)$ and $d(\T'_\ell)$ for any phase $\ell$.
For any client $j\in D_\ell$ and any node-set $S\supseteq N_j,\ r\notin S$, we have
$\sum_{e\in\dt(S)}z_{e,t_{\ell}}\geq\sum_{i\in S,t\leq t_\ell}x_{ij,t}
\geq\sum_{i\in N_j,t\leq\tau_j}x_{ij,t}\geq\frac{1}{2}$.
Therefore, $(2z_{e,t_\ell})$ forms a fractional Steiner tree of cost at most $2t_\ell$ on
the supernodes and $r$, 
and hence, $d(\T''_\ell)\leq 4t_\ell$ since it is well known that the cost of an MST is at
most twice the cost of a fractional solution to the Steiner-tree LP. 
For $j\in\C_\ell$, let $\deg_j$ denote the degree of the cluster $N_j$ in $\T''_\ell$.
Observe that if $e$ is an edge of $\T''_\ell$ joining $N_j$ and $N_k$ (so $j,k\in\C_\ell$), 
then $c_e\geq 24\max\{\bC_j,\bC_k\}$, since 
$30\max\{\bC_j,\bC_k\}\leq c_{jk}\leq 3\bC_j+c_e+3\bC_k$.
So the ($d$-) cost of adding the additional edges to $\T''_\ell$ is at most
$\frac{1}{M}\cdot\sum_{j\in\C_\ell}\deg_j\cdot 3\bC_j\leq d(\T''_\ell)/4$, and hence,
$d(\T_\ell)\leq 5t_\ell$.

Now consider the cost of adding facility edges to $\T_\ell$ in step R2. For each facility
edge $(i,k)$ added, we know that $i\in N_\nbr(k)$, 
$k\in\C_\ell$, and $k$ is assigned to $i$. 
So we have $c_{ik}\leq 9\bC_k$. Observe that
each client $k\in\C_\ell$ is responsible for at most one such facility edge. 
So the $d$-cost of these facility edges is at most
$\frac{1}{M}\cdot\sum_{j\in\C_\ell}9\bC_j\leq
\frac{1}{M}\sum_{j\in\C_\ell}\deg_j\cdot 9\bC_j\leq\frac{3}{4}\cdot d(\T''_\ell)$. 
Thus, $d(\T'_\ell)\leq d(\T_\ell)+\frac{3}{4}\cdot d(\T''_\ell)\leq 8t_\ell$.

Finally, we prove that the latency cost of any client $j$ is at most $64\tau_j=384\bL_j$. 
Let $j$ be assigned to facility $i$. Let $\ell$ be the smallest index
such that $j\in D_\ell$, so $t_\ell\leq 2\tau_j$. We first argue that if $i$ is part of
the tree $\T'_{\ell'}$, then $\ell'\leq\ell$. If $j\in\C$, this follows since we know that
$i\in N_{\nbr(j)}$ and $\ell'=\min\{r:\exists k\in\C_r\text{ with }\nbr(k)=\nbr(j)\}$. If
$j\notin\C$, then we know that $\sg(j)\in\C_\ell$ is also assigned to $i$, and so by the
preceding argument, we again have that $\ell'\leq\ell$. Thus, the latency-cost of $j$ is
bounded by 
$\sum_{r=0}^{\ell'}2d(\T'_r)\leq 16\sum_{r=0}^{\ell'}t_r\leq 32t_{\ell'}\leq 
64\tau_j$. 
\end{proof}

\subsection{\mlufl with a uniform time-metric} \label{mlufl-unif}
We now consider the special case of \mlufl, referred to as {\em uniform \mlufl}, where the
time-metric $d$ is uniform, that is, $d_{ii'}=1$ for all $i,i'\in\F\cup\{r\}$. 
When the connection costs form a metric, we call it the {\em metric uniform \mlufl}.
We consider
the following simpler LP-relaxation of the problem, where 
the time $t$ now ranges from $1$ to $n$. 
\begin{gather*}
\min \quad \sum_{i,t} f_iy_{i,t}+\sum_{j,i,t}(c_{ij}+t)x_{ij,t} 
\qquad \text{subject to} \tag{Unif-P} \label{uniflp} \\
\sum_{i,t} x_{ij,t} \geq 1 \quad \forall j; 
\qquad x_{ij,t} \leq y_{i,t} \quad \forall i,j,t; 
\qquad \sum_i y_{i,t} \leq 1 \quad \forall t; 
\qquad x_{ij,t}, y_{i,t} \geq 0 \quad \forall i,j,t. 
\end{gather*}

Let $(x,y)$ be an optimal solution to \eqref{uniflp}, and $\OPT$ be its
value. Let $\bC_j=\sum_{i,t}c_{ij}x_{ij,t}$, $\bL_j=\sum_{i,t}tx_{ij,t}$.
As stated in the introduction, uniform \mlufl generalizes: 
(i) set cover, when the facility and connection costs are arbitrary; 
(ii) \mssc, when the facility costs are zero (\zfc{} \mlufl);
and (iii) metric \ufl, when the connection costs form a metric. 
We obtain approximation bounds for uniform \mlufl, \zfc{} \mlufl, and metric 
\mlufl (Theorems~\ref{unifthm} and~\ref{metricthm}) that complement
these observations. 

The main result of this section is Theorem~\ref{metricredn}, 
which shows that a $\rho_\mufl$-approximation algorithm for \ufl and a $\gm$-approximation
algorithm for \zfc{} \mlufl (with metric connection costs) can be
combined to yield a $(\rho_\mufl+2\gm)$-approximation algorithm for metric uniform \mlufl.  
Taking $\rho_\mufl=1.5$~\cite{Byrka07} and $\gm=9$ (part (ii) of Theorem~\ref{unifthm}), 
we obtain a 19.5-approximation algorithm. We improve this to $10.773$ by using a more
refined version of Theorem~\ref{metricredn}, which capitalizes on the asymmetric
approximation bounds that one can obtain for different portions of the total cost in \ufl
and \zfc{} \mlufl.  

We note that by considering each $(i,t)$ as a facility, since the connection costs
$c_{ij}+t$ form a metric, one can view metric \mlufl as a variant of metric \ufl, 
and use the ideas in~\cite{BaevRS09} to devise an $O(1)$-approximation for this variant. 
We instead present our alternate algorithm based on the reduction in
Theorem~\ref{metricredn},    
because this reduction is quite robust and versatile. In particular, it allows us to:
(a) handle certain extensions of the problem, e.g., the setting where we have non-uniform
latency costs (see Section~\ref{extn}), for which the above reduction fails since we do
not necessarily obtain metric connection costs, and 
(b) devise algorithms for the uniform latency versions of other facility location
problems, 
where the cost of a facility does not depend on the client-set assigned to it (so one can
assign a client to any open facility freely without affecting the facility costs).   
For instance, consider uniform \mlufl with the restriction that at most $k$ facilities may
be opened: 
our technique yields a $(\rho_{\kmed} + 2\gm)$ approximation for this problem, using a  
$\rho_{\kmed}$-approximation for $k$-median.

\begin{theorem} \label{unifgenthm} \label{zerofacthm} \label{unifthm}
One can obtain:

\noindent
(i) an $O(\ln m)$-approximation algorithm for uniform \mlufl with arbitrary
facility- and connection- costs. 

\noindent
(ii) a solution of cost at most
$\frac{1}{1-\al}\sum_j\bC_j+\frac{4}{\al}\left\lceil\frac{1}{\al}\right\rceil\sum_j\bL_j$
for \zfc \mlufl, for any parameter $\al\in(0,1)$. Thus, setting $\al=\frac{8}{9}$, yields
solution of cost at most $9\cdot\OPT$.  
\end{theorem}

We defer the proof of Theorem~\ref{unifgenthm} to the end of the section,  
and focus first on detailing the aforementioned reduction.

\begin{theorem} \label{metricredn}
Given a $\rho_\mufl$-approximation algorithm $\A_1$ for \ufl, and a $\gm$-approximation
algorithm $\A_2$ for uniform \zfc{} \mlufl, one can obtain a
$(\rho_\mufl+2\gm)$-approximation algorithm for metric uniform \mlufl. 
\end{theorem}

\begin{proof}
Let $\I$ denote the metric uniform \mlufl instance, and $\iopt$ denote the cost of an
optimal integer solution.  
Let $\I_\mufl$ be the \ufl instance obtained form $\I$ by ignoring the latency costs,
and $\I_\mzfc$ be the \zfc{} \mlufl instance obtained from $\I$ by setting all 
facility costs to zero. Let $\iopt_\mufl$ and $\iopt_\mzfc$ denote respectively the cost of
the optimal (integer) solutions to these two instances. Clearly, we have 
$\iopt_\mufl, \iopt_\mzfc\leq\iopt$. 
We use $\A_1$ to obtain a near-optimal solution to $\I_\mufl$: let $F_1$ be the set of
facilities opened and let $\sg_1(j)$ denote the facility in $F_1$ to which client $j$ is
assigned. So we have $\sum_{i\in F_1}f_i+
\sum_jc_{\sg_1(j)j}\leq\rho_\mufl\cdot\iopt_{\mufl}$. 
We use $\A_2$ to obtain a near-optimal solution to $\I_\mzfc$: let 
$F_2$ be the set of open facilities, $\sg_2(j)$ be the facility to which client $j$ is
assigned, and $\pi(i)$ be the position of facility $i$. So we have
$\sum_j\bigl(c_{\sg_2(j)j}+\pi(\sg_2(j))\bigr)\leq\gm\cdot\iopt_{\mzfc}$.

We now combine these solutions as follows.
For each facility $i\in F_2$, let $\mu(i)\in F_1$ denote the facility in $F_1$ that is
nearest to $i$. We open the set $F=\{\mu(i): i\in F_2\}$ of facilities. The position of
facility $i\in F$ is set to $\min_{i'\in F_2:\pi(i')=i}\pi(i')$. Each facility in 
$F$ is assigned a distinct position this way, but some positions may be vacant.
Clearly we can always convert the above into a proper ordering of $F$ where each 
facility $i\in F$ occurs at position $\kp(i)\leq\min_{i'\in F_2:\pi(i')=i}\pi(i')$.
Finally, we assign each client $j$ to the facility $\phi(j)=\mu(\sg_2(j))\in F$.
Note that $\kp(\phi(j))\leq\pi(\sg_2(j))$ (by definition).
For a client $j$, we now have
$c_{\phi(j)j}\leq c_{\sg_2(j)\mu(\sg_2(j))}+c_{\sg_2(j)j}\leq c_{\sg_2(j)\sg_1(j)}+c_{\sg_2(j)j}
\leq c_{\sg_1(j)j}+2c_{\sg_2(j)j}$.
Thus, the total cost of the resulting solution is at most
$\sum_{i\in F_1}f_i+\sum_j\bigl(c_{\sg_1(j)j}+2c_{\sg_2(j)j}+\pi(\sg_2(j))\bigr)
\leq (\rho_{\mufl}+2\gm)\cdot\iopt$.
\end{proof}

We call an algorithm a $(\rho_f,\rho_c)$-approximation algorithm for \ufl if 
given an LP-solution to \ufl with facility- and connection- costs $\fcost$ and $\ccost$
respectively, it returns a solution of cost at most $\rho_f\fcost+\rho_c\ccost$.
Similarly, we say that an algorithm is a $(\gm_c,\gm_l)$-approximation
algorithm for uniform \zfc{} \mlufl if given a solution to \eqref{mlufllp}
with connection- and latency- costs $\ccost$ and $\lcost$ respectively, it returns a
solution of cost at most $\gm_c\ccost+\gm_l\lcost$. 
The proof of Theorem~\ref{metricredn} easily yields the following more general result.

\begin{corollary} \label{redncorr}
One can combine a $(\rho_f,\rho_c)$-approximation algorithm for \ufl, and a 
$(\gm_c,\gm_l)$-approximation algorithm for uniform ZFC \mlufl, to 
obtain a solution of cost at most $\max\{\rho_f,\rho_c+2\gm_c,\gm_l\}\cdot\OPT$.
\end{corollary}
\begin{proof}
The proof mimics the proof of Theorem~\ref{metricredn}. The only new observation is that
$(x,y)$ yields an LP-solution to 
(i) $\I_{\mufl}$ with facility cost $\sum_{i,t}f_iy_{i,t}$ and connection cost
$\sum_{j,i,t}c_{ij}x_{ij,t}$; and
(ii) $\I_{\mzfc}$ with connection cost $\sum_{j,i,t}c_{ij}x_{ij,t}$ and latency cost
$\sum_{j,i,t}tx_{ij,t}$.
Thus, applying the construction in the proof of Theorem~\ref{metricredn} yields a solution
of total cost at most 
$$
\rho_f\sum_{i,t}f_iy_{i,t}+(\rho_c+2\gm_c)\sum_{j,i,t}c_{ij}x_{ij,t}+\gm_l\sum_{j,i,t}tx_{ij,t}
\leq\max\{\rho_f,\rho_c+2\gm_c,\gm_l\}\cdot\OPT. \\[-4ex]
$$ 
\end{proof}

Combining the 
$\bigl(\frac{\ln(1/\beta)}{1-\beta},\frac{3}{1-\beta}\bigr)$-approximation algorithm for
\ufl~\cite{ShmoysTA97} with the  
$\bigl(\frac{1}{1-\al},\frac{4}{\al}\left\lceil\frac{1}{\al}\right\rceil\bigr)$-approximation
algorithm for \zfc{} \mlufl (part (ii) of Theorem~\ref{unifthm}) gives the following result.

\begin{theorem} \label{metricthm}
For any $\al,\beta\in(0,1)$, one can obtain a solution of cost 
$\max\bigl\{\frac{\ln(1/\beta)}{1-\beta},\frac{3}{1-\beta}+\frac{2}{1-\al},
\frac{4}{\al}\bigl\lceil\frac{1}{\al}\bigr\rceil\bigr\}\cdot\OPT$. 
Thus, taking $\al=0.7426,\ \beta=0.000021$, we obtain a $10.773$-approximation algorithm. 
\end{theorem}

\subsubsection{Proof of Theorem~\ref{unifgenthm}}

The following lemma will often come in handy.

\begin{lemma} \label{capviol}
Let $(\hx,\hy)$ be a solution satisfying $\sum_i\hy_{i,t}\leq k$ for every time $t$, where 
$k\geq 0$ is an integer, and all the other constraints of \eqref{uniflp}.
Then, one can obtain a feasible solution $(x',y')$ to \eqref{uniflp} such that 
(i) $\sum_{i,t}f_iy'_{i,t}=\sum_{i,t}f_i\hy_{i,t}$; 
(ii) $\sum_{j,i,t}c_{ij}x'_{ij,t}=\sum_{j,i,t}c_{ij}\hx_{ij,t}$; and
(iii) $\sum_{j,i,t}tx'_{ij,t}\leq k\cdot\sum_{j,i,t}t\hx_{ij,t}$.
\end{lemma}

\begin{proof}
For each time $t$, define $S_t=\{(i,t):\hy_{i,t}>0\}$. 
The idea is to simply ``spread out'' the $S_t$ sets. Let $S_0=\bigcup_t S_t$ be an ordered
list where all the $(\cdot,t)$ pairs are listed before any $(\cdot,t+1)$ pair, and the
pairs for a given $t$ (i.e., $(i,t)\in S_t$) are listed in arbitrary order. 
Let $T_0=\left\lceil\sum_{i,t}\hy_{i,t}\right\rceil$. We divide the pairs in $S_0$ into
$T_0$ groups as follows. Initialize 
$\ell\assign 1$, $S\assign S_0$. 
For a set $A$ of $(i,t)$ pairs, we define the $\hy$-weight of $A$ as
$\hy(A)=\sum_{(i,t)\in A}\hy_{i,t}$, and the $\hx_j$-weight of $A$ as  
$\sum_{(i,t)\in A}\hx_{ij,t}$.  
If $0<\hy(S)\leq 1$, then we set $G_\ell=S$ to end the grouping process.  
Otherwise, group $G_\ell$ includes all pairs of $S$, taken in order starting from the
first pair, stopping when the total $\hy$-weight of the included pairs becomes at least 1; 
all the included pairs are also deleted from $S$. If the $\hy$-weight of $G_\ell$ now
exceeds 1, then we split the last pair $(i,t)$ into two copies:  
we include the first copy in $G_\ell$ and retain the second copy in $S$, and distribute
$\hy_{i,t}$ across the $\hy$-weight of the two copies so that $\hy(G_\ell)$ is now 
{\em exactly} 1. (Thus, the new $\hy$-weight of $S$ is precisely its old $\hy$-weight
$-1$.) Also, for each client $j$, we distribute $\hx_{ij,t}$ across the $\hx_j$-weight of
the two copies arbitrarily while maintaining that the $\hx_j$-weight of each copy is at
most its $\hy$-weight. 
We update 
$\ell\assign\ell+1$, and continue in this fashion with the 
current (i.e., ungrouped) list of pairs $S$. 
Note that an $(i,t)$ pair in $S_0$ may be split into at most two copies above (that lie in 
consecutive groups). To avoid notational clutter, we call both these copies 
$(i,t)$ and use $\hy^{\ell}_{i,t}$ to denote the $\hy$-weight of the copy in group
$G_\ell$ (which is equal to the original $\hy_{i,t}$ if $(i,t)$ is not
split). Analogously, we use $\hx^\ell_{ij,t}$ to denote the $\hx_j$-weight of the copy of   
$(i,t)$ in group $G_\ell$.

For every facility $i$, client $j$, and $\ell=1,\ldots,T_0$, we set 
$y'_{i,\ell}=\sum_{t:(i,t)\in G_\ell}\hy^{\ell}_{i,t}$ and
$x'_{ij,\ell}=\sum_{t:(i,t)\in G_\ell}\hx^\ell_{ij,t}$, so we have 
$x'_{ij,\ell}\leq y'_{i,\ell}$.  
It is clear that $\sum_\ell y'_{i,\ell}=\sum_t \hy_{i,t}$ and  
$\sum_\ell x'_{ij,\ell}=\sum_t \hx_{ij,t}$ for every facility $i$ and client $j$. 
Thus, $(x',y')$ is feasible to \eqref{uniflp}, and parts (i) and (ii) of the lemma hold. 
To prove part (iii), note that if (some copy of) $(i,t)$ is in $G_\ell$, then 
$\ell\leq kt$, since then we have 
$\bigl(\bigcup_{r=1}^{\ell-1} G_r\bigr)\subset\bigcup_{t'=1}^t S_{t'}$ and so
$\ell-1<kt$. 
Thus, for any client $j$, we have 
$$
\sum_{i,\ell}\ell x'_{ij,\ell}=
\sum_{i,\ell}\ell\bigl(\sum_{t:(i,t)\in G_\ell}\hx^\ell_{ij,t}\bigr)=
\sum_{i,t}\sum_{\ell:(i,t)\in G_\ell}\ell\hx^\ell_{ij,t}\leq
\sum_{i,t}kt\sum_{\ell:(i,t)\in G_\ell}\hx^\ell_{ij,t}=k\cdot\sum_{i,t}t\hx_{ij,t}.
$$
\end{proof}

\begin{proofof}{part (i) of Theorem~\ref{unifgenthm}}
We round the LP-optimal solution $(x,y)$ by using filtering followed by standard
randomized rounding. 
Clearly, we may assume that $\sum_{i',t}x_{i'j,t}=1$ and $y_{i,t}=\max_j x_{ij,t}$ for
every $i,j,t$. Also, we may assume that if $\sum_i y_{i,t+1}>0$, then $\sum_i y_{i,t}=1$,
because otherwise for some facility $i$ and some $\e>0$, we may decrease $y_{i,t+1}$ by
$\e$ and increase $y_{i,t}$ by $\e$, and modify the $\{x_{ij,t+1}, x_{ij,t}\}_j$ values
appropriately so as to maintain feasibility, without increasing the total cost. 
Define $N_j=\{(i,t):c_{ij}+t\leq 2(\bC_j+\bL_j)\}$ for a 
client $j$. The algorithm is as follows.

\begin{list}{U\arabic{enumi}.}{\usecounter{enumi} \addtolength{\leftmargin}{-1ex}}
\item For each $(i,t)$, we set $Y_{i,t}=1$ independently with probability
$\min\{4\ln m\cdot y_{i,t},1\}$. 

\item Considering each client $j$, if $\{(i,t)\in N_j: Y_{i,t}=1\}=\es$, then set
$Y_{i_j,1}=1$ where $i_j$ is such that $f_{i_j}=\min_{(i,t)\in N_j}f_i$ (note that 
$(i_j,1)\in N_j$). Let $S_j=\{(i,t)\in N_j: Y_{i,t}=1\}$ (which is non-empty).
Assign each client $j$ to the $(i,t)$ pair in $S_j$ with minimum $c_{ij}+t$ value,
i.e., set $X_{ij,t}=1$.  
 
\item Let $K=\max_t \sum_i Y_{i,t}$. 
Use Lemma~\ref{capviol} to convert $(X,Y)$ into a feasible integer solution to
\eqref{uniflp}. 
\end{list}

Let $C_j$ and $L_j$ denote respectively the connection cost and latency cost of
client $j$ in $(X,Y)$. We argue that 
(i) $\E{\sum_{i,t}f_iY_{i,t}}=O(\ln m)\cdot\sum_{i,t}f_iy_{i,t}$,  
(ii) with probability 1, $C_j+L_j\leq 2(\bC_j+\bL_j)$ for every client $j$, and  
(iii) $K=O(\ln m)$ with high probability, and in expectation. 
The theorem then follows from Lemma~\ref{capviol}.

For any client $j$, we have $\sum_{(i,t)\in N_j}x_{ij,t}\geq\frac{1}{2}$ (by Markov's
inequality). Thus, $f_{i_j}\leq 2\sum_{(i,t)\in N_j}f_iy_{i,t}$ and 
$\Pr[\sum_{(i,t)\in N_j}Y_{i,t}=0]$ is at most $e^{-4\ln m\cdot\frac{1}{2}}=1/m^2$.
The expected cost of opening facilities in step U1 is clearly at most 
$4\ln m\cdot\sum_{i,t}f_iy_{i,t}$. The expected facility-opening cost in step U2 is at
most $\Pr[\text{facility is opened in step U2}]\cdot\sum_j f_{i_j}
\leq m\cdot\frac{1}{m^2}\cdot 2m\sum_{i,t}f_iy_{i,t}$. 
So $\E{\sum_{i,t}f_iY_{i,t}}=O(\ln m)\cdot\sum_{i,t}f_iy_{i,t}$.
Since we always open some $(i,t)$ pair in $N_j$, we have 
$C_j+L_j\leq 2(\bC_j+\bL_j)$ for 
every client $j$. 

Let $\nt=\{t:\sum_i y_{i,t}>0\}$. Note that
$|\nt|\leq 1+\sum_{i,t}y_{i,t}\leq 1+\sum_{j,i,t}x_{ij,t}=m+1$.
After step U1, we have $\E{\sum_i Y_{i,t}}\leq 4\ln m$ for every time $t\in\nt$. 
Since the $Y_{i,t}$ random variables are independent, we also have 
$\Pr[\sum_i Y_{i,t}>8\ln m]\leq 1/m^2$ for all $t\in\nt$ 
(and also, $\E{\max_{t\in\nt}\sum_i Y_{i,t}}=O(\ln m)$). Thus, after step U2,
we have $\Pr[\sum_i Y_{i,1}>8\ln m]\leq 1/m^2+1/m$ and 
$\Pr[\sum_i Y_{i,t}>8\ln m]\leq 1/m^2$ for all $t\in\nt,\ t>1$. 
Hence, $\Pr[K>8\ln m]\leq 2/m$. 
(This also shows that $\E{K}=O(\ln m)$.)
\end{proofof}

\begin{proofof}{part (ii) of Theorem~\ref{unifgenthm}}
We round $(x,y)$ by applying filtering~\cite{LinV92} followed by Lemma~\ref{capviol} to
reduce the problem to a \mssc problem, and then use the result of Feige et
al.~\cite{FeigeLT04} to obtain a near-optimal solution to this \mssc problem. Let 
\begin{equation}
\min \ \ \sum_{j,t}tx_{j,t} \quad \mathrm{s.t.} \quad
\sum_{t} x_{j,t} \geq 1 \ \ \forall j, \quad 
x_{j,t}\leq\sum_{S:j\in S}y_{S,t} \ \ \forall j,t, \quad
\sum_S y_{S,t}\leq 1\ \ \forall t, \quad x,y\geq 0. \tag{P1} \label{mssclp}
\end{equation}
denote the standard LP-relaxation of \mssc~\cite{FeigeLT04} (here $j$ indexes the
elements, $S$ indexes the sets, and $t$ indexes time).   
Feige et al. showed that given a solution $(\hx,\hy)$ to \eqref{mssclp}, one can obtain in
polytime an integer solution of cost at most $4\cdot\sum_{j,t}t\hx_{j,t}$.

The rounding algorithm for \zfc{} \mlufl is as follows.
Define $N_j=\{i: c_{ij}\leq \bC_j/(1-\al)\}$, so $\sum_{i\in N_j,t}x_{ij,t}\geq\al$. 
For all $i,j,t$, set $\hy_{i,t}=y_{i,t}/\al$, and $\hx_{ij,t}=x_{ij,t}/\al$ if $i\in N_j$
and $\hx_{ij,t}=0$ otherwise. It is easy to see that $(\hx,\hy)$ satisfies 
$\sum_i \hy_{i,t}\leq\frac{1}{\al}$ for all $t$, and all the other constraints of
\eqref{uniflp}. We use Lemma~\ref{capviol} to convert $(\hx,\hy)$ to a feasible solution
$(x',y')$ to \eqref{uniflp}. 
Next, we extract a solution to \eqref{mssclp} from $(x',y')$. We identify facility $i$ with
the set $\{j:i\in N_j\}$, and set $\bx_{j,t}=\sum_{i\in N_j}x'_{ij,t}$. Now $(\bx,y')$ is
a feasible solution to \eqref{mssclp}. 
Finally, we round $(\bx,y')$ to an integer solution.
This yields the ordering $\ty=(\ty_{i,t})$ of the facilities. For each client $j$, if $j$
is first covered by set $i$ (so $i\in N_j$) at time (or position) $t$ in the \mssc
solution, then we set $\tx_{ij,t}=1$.

\vspace{-2ex}
\paragraph{Analysis.}
Since a client $j$ is always assigned to a facility in $N_j$, the connection cost of $j$
is bounded by $\bC_j/(1-\al)$.
To bound the latency cost, first we bound the cost of $(\hx,\hy)$. Since we modify the
assignment of a client $j$ by transferring weight from farther facilities to nearer ones,
it is clear that $\sum_{i,t}c_{ij}\hx_{ij,t}\leq\sum_{i,t}c_{ij}x_{ij,t}$. Also, clearly
$\sum_{j,i,t}t\hx_{ij,t}\leq\frac{1}{\al}\cdot\sum_{j,i,t}tx_{ij,t}$. 
and $\sum_i\hy_{i,t}\leq\frac{1}{\al}$. Thus, applying Lemma~\ref{capviol} yields
$(x',y')$ satisfying 
$\sum_{j,i,t}tx'_{ij,t}\leq
\left\lceil\frac{1}{\al}\right\rceil\frac{1}{\al}\cdot\sum_{j,i,t}tx_{ij,t}$.
The result of~\cite{FeigeLT04} now implies that \linebreak
$\sum_{j,i,t}t\tx_{ij,t}\leq 4\sum_{j,t}t\bx_{j,t}=4\sum_{j,i,t}tx'_{ij,t}
\leq\frac{4}{\al}\left\lceil\frac{1}{\al}\right\rceil\sum_j tx_{ij,t}$.
\end{proofof}

\section{LP-relaxations and algorithms for the minimum-latency problem} \label{lpml} 
In this section, we consider the minimum-latency (\ml) problem and apply our techniques 
to obtain LP-based insights and
algorithms for this problem. We give two LP-relaxations for \ml with constant integrality
gap. The first LP \eqref{mllp1} is a specialization of \eqref{mlufllp} to \ml, and to
bound its integrality gap, we only need the fact that the   
natural subtour elimination LP for \tsp  has constant integrality gap. 
The second LP \eqref{mllp2} has exponentially-many variables, one for every path (or tree)
of a given length bound, and 
the separation oracle for the dual problem 
corresponds to an (path- or tree-) orienteering problem. We prove that even a bicriteria  
approximation for the orienteering problem yields an 
approximation for \ml while losing a constant factor. 
(The same relationship also holds between \mgl and ``group orienteering''.)  
As mentioned in the Introduction, we
believe that our results shed new light on \ml and  opens up \ml to new
venues of attack. Moreover, as shown in Section~\ref{extn} these LP-based techniques
can easily be used to handle more general variants of \ml, e.g., $k$-route \ml with
$\Lf_p$-norm latency-costs (for which we give the first approximation algorithm).
We believe that our LP-relaxations are in fact (much) better than what we have accounted
for, and conjecture that the integrality gap of both \eqref{mllp1} and \eqref{mllp2} 
is at most 3.59, which is the currently best known approximation factor for \ml.

Let $G=(\D\cup\{r\},E)$ be the complete graph on $N=|\D|+1$ nodes with edge
weights $\{d_e\}$ that form a metric. Let $r$ be the root node at which the path
visiting the nodes must originate. We use $e$ to index $E$ and $j$ to index the nodes. In
both LPs, we have variables $x_{j,t}$ for $t\geq d_{jr}$ to denote if $j$ is visited at
time $t$ (where $t$ ranges from $1$ to $\Time$); for convenience, we think of
$x_{j,t}$ as being defined for all $t$, with $x_{j,t}=0$ if $d_{jr}>t$. 
(As in Section~\ref{mlufl-gen}, one can move to a polynomial-size LP losing a
$(1+\e)$-factor.)   

\vspace{-2ex}
\paragraph{A compact LP.} 
As before, we use a variable $z_{e,t}$ to denote if $e$ has been traversed by 
time $t$. 
\begin{gather*}
\min \quad \sum_{j,t} tx_{j,t} \qquad  \text{subject to} \tag{LP1} \label{mllp1} \\
\sum_{t} x_{j,t} \ge 1 \quad \forall j; 
\qquad\ \sum_{e} d_ez_{e,t} \le t \quad \forall t; 
\qquad\ \sum_{e\in \delta(S)} z_{e,t} \ge \sum_{t'\le t} x_{j,t'} 
\quad \forall t, S\subseteq \D, j; 
\qquad x, z \ge 0. 
\end{gather*}

\vspace{-2ex}
\begin{theorem} \label{mllp1thm}
The integrality gap of \eqref{mllp1} is at most $10.78$.
\end{theorem}
\vspace{-1ex}
\begin{proof}
Let $(x,z)$ be an optimal solution to \eqref{mllp1}, and $\bL_j=\sum_t tx_{j,t}$. 
For $\al\in[0,1]$, define the {\em $\al$-point of $j$}, $\tau_j(\al)$, to be the smallest $t$  
such that $\sum_{t'\leq t}x_{jt'}\geq\al$. Let $D_t(\al)=\{j:\tau_j(\al)\leq t\}$.
We round $(x,z)$ as follows. We pick $\al\in(0,1]$ according to the density function
$q(x)=2x$. At each time $t$, we utilize the $\frac{3}{2}$-integrality-gap of the
subtour-elimination LP for TSP and the parsimonious property
(see~\cite{Wolsey80,ShmoysW90,GoemansB90,BienstockGSW93}), to round $\frac{2z}{\al}$ and
obtain a tour on $\{r\}\cup D_t(\al)$ of cost 
$C_t(\al)\leq\frac{3}{\al}\cdot\sum_e d_ez_{e,t}\leq \frac{3t}{\al}$. 
We now use Lemma~\ref{gk} to combine these tours. 

\vspace{-1ex}
\begin{lemma}[\cite{GK} paraphrased] \label{gk}
Let $\Tour_1,\ldots,\Tour_k$ be tours containing $r$, with $\Tour_i$ having cost $C_i$ and
containing $N_i$ nodes, where $N_0:=1\leq N_1\leq \ldots\leq N_k=N$. One can find a
subset $\Tour_{i_1},\ldots,\Tour_{i_b=k}$ of tours, and a way of concatenating them 
that gives total latency at most $\frac{3.59}{2}\sum_i C_i(N_i-N_{i-1})$.
\end{lemma}

\noindent
\vspace{-1ex}
The tours we obtain for the different times are nested (as the $D_t(\al)$s
are nested). 
So 
$\sum_{t\geq 1} C_t(\al)(|D_t(\al)|-|D_{t-1}(\al)|)=
\sum_j\sum_{t:j\in D_t(\al)\sm D_{t-1}(\al)}C_t(\al)=\sum_j C_{\tau_j(\al)}(\al)\leq 
3\sum_j\frac{\tau_j(\al)}{\al}$. Thus, using Lemma~\ref{gk}, and taking expectation over
$\al$ (note that $\E{\frac{\tau_j(\alpha)}{\alpha}} \le 2\bL_j$), we obtain that the total
latency-cost is at most $(3.59\cdot 3)\sum_j\bL_j$.
\end{proof}

\noindent
Note that in the above proof we did not need any procedure to solve $k$-MST or
its variants, but rather just needed the integrality gap for the subtour-elimination LP to
be a constant. Also, we can modify the rounding procedure to ensure that 
$\text{(latency-cost of $j$)}\leq 18\tau_j(0.5)$ for 
each client $j$, as follows. (Such a guarantee is useful to bound the total cost when its
measured as an ${\mathcal L}_p$ norm for $p > 1$; see Section \ref{extn}.)
We now only consider times $t_\ell=2^\ell$. 
Recall that for any $\al\in(0,1)$, at each time $t$, we can obtain a tour on $\{r\}\cup
D_t(\al)$ of cost $C_t(\al)\leq\frac{3t}{\al}$. 
We take this tour for $t_\ell$, and traverse the resulting tour randomly clockwise or
anticlockwise (this choice can easily be derandomized), 
and concatenate all these tours. 
Let $\ell_j(\al)$ be the smallest $\ell$ such that $t_\ell\geq \tau_j(\al)$. 
So the (expected) latency-cost of $j$ is at most 
$\sum_{\ell<\ell_j(\al)}\frac{3t_\ell}{\al}+\frac{1}{2}\cdot\frac{3t_{\ell_j(\al)}}{\al}
\leq \frac{4.5t_{\ell_j(\al)}}{\al} \leq \frac{9\tau_j(\al)}{\al}$. 
Fixing $\al=0.5$, we obtain a 36-approximation with the per-client guarantee that
$\text{(latency-cost of $j$)}\leq 18\tau_j(0.5)$ for each $j$.

\paragraph{An exponential-size LP: relating the orienteering and latency problems.}
Let $\Pcol_t$ and $\Tcol_t$ denote respectively the collection of all (simple) paths and
trees rooted at $r$ of length at most $t$. 
For each path $P\in\Pcol_t$, we introduce a variable $z_{P,t}$ that indicates if $P$ is  
the path used to visit the nodes with latency-cost at most $t$.

\vspace{-5pt}
\noindent \hspace*{-3ex}
\begin{minipage}[t]{.49\textwidth}
\begin{alignat}{3}
\min & \quad & \sum_{j,t}tx_{j,t} & \tag{LP2$_\Pcol$} \label{mllp2} \\
\text{s.t.} && \sum_t x_{j,t} & \geq 1 \qquad && \forall j \notag \\ 
&& \sum_{P\in\Pcol_t} z_{P,t} & \leq 1 \qquad && \forall t \label{onep} \\ 
&& \sum_{P\in\Pcol_t: j\in P}z_{P,t} & \geq \sum_{t'\leq t}x_{j,t'} 
\qquad && \forall j,t \label{jcov} \\
&& x, z & \geq 0. \notag
\end{alignat}
\end{minipage}
\quad \rule[-29ex]{1pt}{26ex}
\begin{minipage}[t]{.49\textwidth}
\begin{alignat}{3}
\max & \quad & \sum_j\al_j & - \sum_t\beta_t \tag{LD2} \label{dlp2} \\
\text{s.t.} && \al_j & \leq t+\sum_{t'\geq t}\tht_{j,t'} \qquad && \forall j,t
\label{dineq1} \\
&& \sum_{j\in P}\tht_{j,t} & \leq \beta_t \qquad && \forall t, P\in\Pcol_t 
\label{dineq2} \\ 
&& \al, \beta, \tht & \geq 0. \label{dineq3}
\end{alignat}
\end{minipage}

\medskip \noindent
\eqref{onep} and \eqref{jcov} encode that at most one path may be chosen for any time $t$,
and that every node $j$ visited at time $t'\le t$ must lie on this path.
\eqref{dlp2} is the dual LP with 
exponentially many constraints. Let $(\text{LP2}_\T)$ be the analogue of \eqref{mllp2}
with tree variables, where we have variables $z_{Q,t}$ for every $Q\in\T_t$, and we 
replace all occurrences of $z_{P,t}$ in \eqref{mllp2} with $z_{Q,t}$.

Separating over the constraints \eqref{dineq2} involves solving a (rooted)
path-orienteering problem: for every $t$, given rewards $\{\tht_{j,t}\}$, we want to
determine if there is a path $P$ rooted at $r$ of length at most $t$ that gathers reward
more than $\beta_t$. 
A $(\rho,\gm)$-\{path, tree\} approximation algorithm for the path-orienteering problem is
an algorithm that always returns a \{path, tree\} rooted at $r$ of length at most
$\gm(\text{length bound})$ that gathers reward at least $(\text{optimum reward})/\rho$. 
Chekuri et al.~\cite{ChekuriKP08} give a $(2+\e,1)$-path approximation 
algorithm, whereas 
\cite{CGRT} design a $(1+\e,1+\e)$-tree approximation for orienteering (note that weighted 
orienteering can be reduced to unweighted orienteering with a $(1+\e)$-factor loss).    
We prove that even a $(\rho,\gm)$-tree approximation algorithm for orienteering can be
used to obtain an $O(\rho\gm)$-approximation for \ml. 
First, we show how to compute a near-optimal LP-solution. Typically, one 
argues that, scaling the solution computed by the ellipsoid method run on the dual with the
approximate separation oracle yields a feasible and near-optimal dual solution, and this 
is then used to obtain a near-optimal primal solution (see, e.g.,~\cite{JainMS03}). 
However, in our case, we have negative terms in the dual objective function, which makes
our task trickier: if our (unicriteria) $\rho$-approximate separation oracle determines
that $(\al,\beta,\tht)$ is feasible, then although $(\al,\rho\beta,\tht)$ is feasible to
\eqref{dlp2}, one has {\em no guarantee on the value of this dual solution}. 
Instead, the notion of approximation we obtain for the primal solution computed involves
bounded violation of the constraints. 

\newcommand{\mlpath}[1]{\ensuremath{\bigl(\text{LP2}_{\Pcol}^{({#1})}\bigr)}}
\newcommand{\mltree}[1]{\ensuremath{\bigl(\text{LP2}_{\Tcol}^{({#1})}\bigr)}}

Let $\mlpath{a,b}$ be \eqref{mllp2} where
we replace $\Pcol_t$ by $\Pcol_{bt}$, and the RHS of \eqref{onep} is now $a$. Let
$\mltree{a,b}$ be defined analogously.
Let $\OPT_\Pc$ be the optimal value of \eqref{mllp2} (i.e., $\mlpath{1,1}$). 
Note that $\OPT_\Pc$ is a lower bound on the optimum latency.

\begin{lemma} \label{apxsep}
Given a $(\rho,\gm)$-tree approximation for the orienteering problem, one can compute a 
feasible solution $(x,z)$ to $\mltree{\rho,\gm}$ of cost at most $\OPT$.
\end{lemma}
\begin{proof}{Lemma~\ref{apxsep}}
Define 
\begin{eqnarray*}
\Pfeas(\nu;a,b) & := &  
\Bigl\{(\al,\beta,\tht):\ \ \eqref{dineq1}, \quad \eqref{dineq3}, \quad
\sum_{j\in P}\tht_{j,t}\leq \beta_t \quad \forall P\in\Pcol_{bt}, \quad 
\sum_j\al_j-a\sum_t\beta_t\geq\nu\Bigr\}, \\ 
\Tfeas(\nu;a,b) & := & 
\Bigl\{(\al,\beta,\tht):\ \ \eqref{dineq1}, \quad \eqref{dineq3}, \quad
\sum_{j\in Q}\tht_{j,t}\leq \beta_t \quad \forall Q\in\Tcol_{bt}, \quad
\sum_j\al_j-a\sum_t\beta_t\geq\nu\Bigr\}.
\end{eqnarray*}
Note that $OPT_\Pcol$ is the largest value of $\nu$ such that $\Pfeas(\nu;1,1)$ is feasible.
Given $\nu, (\al,\beta,\tht)$, if there was an algorithm to either show $(\al,\beta,\tht)\in \Pfeas(\nu;1,1)$
or exhibit a separating hyperplane, then using the ellipsoid method we could optimally solve for $OPT_\Pcol$.
However, such a separation oracle would solve orienteering exactly. 

We use the $(\rho,\gamma)$-tree approximation algorithm to give an {\em approximate}
separation oracle in the following sense. Given $\nu, (\al,\beta,\tht)$, we either show
$(\alpha,\rho\beta,\tht) \in \Pfeas(\nu;1,1)$, 
or we exhibit a hyperplane separating $(\al,\beta,\tht)$ and $\Tfeas(\nu;\rho,\gamma)$. Note that 
$\Tfeas(\nu;\rho,\gamma) \subseteq \Pfeas(\nu;1,1)$. Thus, for a fixed $\nu$, the ellipsoid method
in polynomial time, either certifies that $\Tfeas(\nu;\rho,\gamma)$ is empty or 
returns a point $(\alpha,\beta,\theta)$ with $(\alpha,\rho\beta,\theta) \in \Pfeas(\nu;1,1)$.
Let us describe the approximate separation oracle first, and then use the above fact to prove the lemma. 
First, check if $(\sum_j\al_j-\rho\beta_t\geq\nu)$,
\eqref{dineq1}, and \eqref{dineq3} hold, and if not, we use the appropriate inequality as the 
separating hyperplane between $(\al,\beta,\tht)$ and $\Tfeas(\nu;\rho,\gamma)$. 
Next, for each $t$, we run the $(\rho,\gm)$-tree approximation on the orienteering
problem specified by $\bigl(G,\{d_e\}\bigr)$, root $r$, rewards $\{\tht_{j,t}\}$, and
budget $t$. If for some $t$, we obtain a tree $Q\in\Tcol_{\gm t}$ with reward greater than
$\beta_t$, then we return $\sum_{j\in Q}\tht_{j,t}\leq\beta_t$ as the separating hyperplane.
If not, then for all paths $P$ of length at most $t$ in $G$, we have $\sum_{j\in P}\tht_{j,t} \le \rho\beta_t$ and thus
 $(\alpha,\rho\beta,\theta) \in \Pfeas(\nu;1,1)$.

We find the largest $\nu^*$ (via binary search) such that the ellipsoid method run for
$\nu^*$ with our separation oracle returns a solution $(\al^*,\beta^*,\tht^*)$ with
$(\al^*,\rho\beta^*,\tht^*)\in\Pfeas(\nu^*;1,1)$; hence, we have $\nu^*\leq\OPT_\Pcol$ 
(by duality). Now for $\e>0$, the ellipsoid method run for $\nu^*+\e$ terminates in
polynomial time certifying the infeasibility of $\Tfeas(\nu^*+\e;\rho,\gamma)$.
That is, it generates a polynomial number of inequalities 
of the form \eqref{dineq1}, \eqref{dineq3}, and $(\sum_{j\in Q}\tht_{j,t}\leq\beta_t)$
where $Q\in\Tcol_{\gm t}$, which together with the inequality
$\sum_j\al_j-\rho\sum_t\beta_t\geq\nu^*+\e$ constitute an infeasible system.
Applying Farkas' lemma, equivalently, we get a polynomial sized solution
$(x,z)$ to $\mltree{\rho,\gm}$ that has cost at most $\nu^*+\e$. Taking $\e$ small enough (something
like $1/\exp(\text{input size})$ so that $\ln(1/\e)$ is still polynomially bounded), this
also implies that $(x,z)$ has cost at most $\nu^*\leq\OPT_\Pcol$. This completes the 
proof of the lemma.
\end{proof}

\begin{theorem} \label{mllp2thm}
(i) A feasible solution $(x,z)$ to 
$\mltree{\rho,\gm}$ (or the corresponding LP-relaxation for \mgl) can be rounded to
obtain a solution of expected cost at most $O(\rho\gm)\cdot\sum_{j,t}tx_{j,t}$; \  
(ii) A feasible solution $(x,z)$ to 
$\mlpath{\rho,\gm}$ can be rounded to obtain a solution of cost at most 
$(3.59\cdot 2)\rho\gm\sum_{j,t}tx_{j,t}$. 
\end{theorem}
\begin{proof}
We prove part (i) first.
(Note that for \ml, the analysis leading to Theorem~\ref{mllp1thm} already implies 
that one can obtain a solution of cost at most $10.78\rho\gm\sum_{j,t}tx_{j,t}$, because
setting $\z_{e,t}=\sum_{Q\in\Tcol_{\gm t}:e\in Q}z_{Q,t}$ yields a solution $(x,\z)$ that
satisfies $\sum_e d_e\z_{e,t}\leq\rho\gm t$ and all other constraints of \eqref{mllp1}.) 
We sketch a (randomized) rounding procedure that also works for \mgl and yields improved
guarantees. We may assume that each $z_{Q,t}\in[0,1]$. 
At each time $t_\ell=2^\ell,\ \ell=\ceil{\log_2\Time+4\log_2 m}$, we 
select at most $\ceil{w_{t_\ell}}$ trees from $\Tcol_{bt_\ell}$, picking each
$Q\in\Tcol_{\gm t_\ell}$ with probability $z_{Q,t_\ell}$. 
(We can always do this (efficiently) since for every time $t$, the polytope  
$\{\z\in[0,1]^{\Tcol_{bt}}: \sum_Q \z_{Q,t}\leq \ceil{\sum_Q z_{Q,t}}\}$ is integral.)
We take the union of all these trees. Note that expected cost of the resulting subgraph is
at most $w_t\gm t_\ell\leq\rho\gm t_\ell$. We ``Eulerify'' the resulting subgraph to
obtain a tour for $t_\ell$ of cost at most $2\rho\gm t_\ell$, 
and concatenate these tours. 
The probability that $j$ is not visited (or covered) by the tour for $t_\ell$ is at most 
$1-\sum_{t\leq t_{\ell}}x_{j,t}$, which implies that 
with high probability, we obtain a tour spanning all nodes.  
Letting $\tau_j=t_j\bigl(\frac{2}{3}\bigr)$, the expected latency-cost of $j$ is at most 
$2\rho\gm\bigl(2\tau_j+2\tau_j\sum_{k\geq 0}(\frac{2}{3})^k\bigr)\leq 16\rho\gm\tau_j$. 

\medskip
To prove part (ii), we adopt a rounding procedure that again utilizes Lemma~\ref{gk}, and
yields the stated bound (deterministically).
Recall that $\tau_j(\al)$ denotes the $\al$-point of $j$. 
Let $D_t(\al)=\{j:\tau_j(\al)\leq t\}$.
For any $\al\in(0,1)$, and time $t$, we now show to obtain a tour spanning
$D_t(\al)\cup\{r\}$ of cost at most $\frac{2\rho\gm t}{\al}$. We can then proceed as in  
the rounding procedure for Theorem~\ref{mllp1thm} to
argue that, for the tours we obtain, the quantity $\sum_i C_i(N_i-N_{i-1})$ appearing in
Lemma~\ref{gk} is bounded by $2\rho\gm\sum_j\frac{t_j(\al)}{\al}$.
Hence, choosing $\al$ as before according to the distribution $q(x)=2x$ and taking 
expectations, we obtain a solution with the stated cost. 

Let $K$ be such that 
$K\al,\ \{Kz_{P,t}\}_{P\in\Pcol_{\gm t}}$ are integers. 
For each $P$ with $z_{P,t}>0$, we create $Kz_{P,t}$ copies of each edge on $P$, and direct
the edges away from $r$. Let $A_{P,t}$ denote the resulting arc-set.
Note that in $A_t:=\biguplus_{P:z_{P,t}>0}A_{P,t}$, every node $j\in\D$ has in-degree at
least its out-degree,  
and there are $K\al$ arc-disjoint paths from $r$ to $j$ for each 
$j\in D_t(\al)$. So applying Theorem 2.6 in Bang-Jensen et al.~\cite{BangjensenFJ95}, one
can obtain $K\al$ arc-disjoint out-arborescences rooted at $r$, each containing all 
nodes of $D_t(\al)$. 
Thus, if we pick the cheapest such arborescence and ``Eulerify'' it, 
we obtain a tour spanning $D_t(\al)\cup\{r\}$ of cost at most 
$2\cdot K\rho\gm t\cdot\frac{1}{K\al}=\frac{2\rho\gm t}{\al}$.
\end{proof}

In Appendix~\ref{append-lpml}, we prove an analogue of Lemma~\ref{apxsep} for \mgl. 
Combined with part (i) of Theorem~\ref{mllp2}, this shows that (even) a bicriteria   
approximation for ``group orienteering'' yields an approximation for \mgl while losing a  
constant factor.

\section{Extensions} \label{extn}
We now consider various well-motivated extensions of \mlufl, and show that our LP-based  
techniques and algorithms are quite versatile and extend with minimal effort to yield
approximation guarantees for these more general \mlufl problems. 
Our goal here is to emphasize the flexibility afforded by our LP-based techniques, and we
have not attempted to optimize the approximation factors. 

\vspace{-2ex}
\paragraph{Monotone latency-cost functions with bounded growth, and higher 
$\Lf_p$ norms.}  
Consider the generalization of \mlufl, where we have a non-decreasing function $\ld(.)$
and the latency-cost of client $j$ is given by 
$\ld(\text{time taken to reach the facility serving $j$})$; the goal, as before, is to
minimize the sum of the facility-opening, client-connection, and client-latency costs. 
Say that $\ld$ has growth at most $p$ if $\ld(cx)\leq c^p\ld(x)$ for all 
$x\geq 0,\ c\geq 1$.
It is not hard to see that for concave $\ld$, we obtain the same performance guarantees as
those obtained in Section~\ref{approx} (for $\ld(x)=x$).
So we focus on the case when $\ld$ is convex, 
and obtain an 
$O\bigl(\max\{(p\log^2 n)^p,p\log n\log m\}\bigr)$-approximation algorithm for convex
latency functions of growth $p$. As a {\em corollary}, we obtain an 
$O\bigl(p\log n\max\{\log n,\log m\}\bigr)$-approximation for {\em $\Lf_p$-\mlufl},
where we seek to minimize the facility-opening cost + client-connection cost + 
the $\Lf_p$-norm of client-latencies. 

\begin{theorem} \label{pgrowththm}
There is an $O\bigl(\max\{(p\log^2 n)^p,p\log n\log m\}\bigr)$-approximation algorithm for 
\mlufl with convex monotonic latency functions of growth (at most) $p$.
\end{theorem}

\begin{proof}
We highlight the changes to the algorithm and analysis in Section~\ref{mlufl-gen}. 
We again assume that $\Time=\poly(m)$ for convenience. This assumption can be dropped by 
proceeding as in Theorem~\ref{scalethm}; we do after proving the theorem.
The objective of \eqref{mlufllp} now changes to 
$\min\ \sum_{i,t} f_iy_{i,t}+\sum_{j,i,t}\bigl(c_{ij}+\ld(t)\bigr)x_{ij,t}$.
The {\em only} change to Algorithm~\ref{genalg} is that we now define
$t_\ell=\min\{2^{\ell/p},\Time\}$ and set 
$\Nc:=\ceil{p\log_2(2^{1/p}\tau_{\max})+4\log_2 m}=O(p\log m)$.
Define $\LC_j=\sum_{i,t}\ld(t)x_{ij,t}$. 
Note that we now have $\ld(\tau_j)\leq 12\LC_j$.
Define $\ell_j$ to be the first phase $\ell$ such that $t_\ell\geq\tau_j$. 
Let the random variable $P_j$ be as defined in Lemma~\ref{lbound}.
The failure probability of the algorithm is again at most $1/\poly(m)$. 
The facility-cost incurred in $O(p\log n\log m)\sum_{i,t}f_iy_{i,t}$, and
the connection cost of client $j$ is at most $4\bC_j$. We generalize Lemma~\ref{lbound}
below to show that $\E{L_j}=O\bigl((p\log^2 n)^p\bigr)\ld(t_{\ell_j})\leq 
O\bigl((2p\log^2 n)^p\bigr)\cdot\ld(\tau_j)$, which yields the desired approximation.
We have 
$$
L_j\leq\ld\bigl(\sum_{\ell\leq P_j}d(\Tour_\ell)\bigr)\leq 
\ld\bigl(O(\log^2 n)\sum_{\ell\leq P_j}t_\ell\bigr) 
\leq O(\log^{2p} n)\ld\bigl(\sum_{\ell\leq P_j}t_\ell\bigr),
$$ so
$\E{L_j}\leq O(\log^{2p} n)\Bigl[\ld\bigl(\sum_{\ell\leq\ell_j}t_\ell\bigr)+ 
\sum_{k\geq 1}\Pr[P_j\geq\ell_j+k]\ld\bigl(\sum_{\ell\leq\ell_j+k}t_\ell\bigr)\Bigr]$.
Now, $\Pr[P_j\geq\ell_j+k]\leq\bigl(\frac{4}{9}\bigr)^k$, 
$\sum_{\ell\leq\ell_j+k}t_\ell\leq t_{\ell_j}\cdot\frac{2^{k/p}}{1-2^{-1/p}}$, and 
$(2^{1/p}-1)\geq\frac{\ln 2}{p}$. 
Plugging these in gives, 
$$
\E{L_j} = O(\log^{2p} n)\cdot\frac{2}{(2^{1/p}-1)^p}\cdot\ld(t_{\ell_j})\cdot
\sum_{k\geq 0}\Bigl(\frac{4}{9}\Bigr)^k2^{k}
= O\bigl((p/\ln 2)^p\log^{2p} n\bigr)\ld(t_{\ell_j}). 
$$

\vspace{-7ex}
\end{proof}

\paragraph{Removing the assumption \boldmath $\Time=\poly(m)$ in Theorem~\ref{pgrowththm}.}
As in the case of Theorem~\ref{scalethm}, to drop the assumption that $\Time=\poly(m)$, we  
(a) solve the LP considering only times in $\TS=\{\Time_0,\ldots,\Time_k\}$ 
(where $\Time_r=\ceil{(1+\e)^r}$); 
and (b) set $t_\ell=\TS(\bbL\cdot 2^{\ell/p})$ and the number of phases to 
$\Nc:=\ceil{p\log_2(2^{1/p}\tau_{\max}/\bbL)+4\log_2 m}$, 
where $\bbL=(\sum_{j,i,t} tx_{ij,t})/m$.  
Note that $\Nc=O(p\log m)$, and 
$\ld(\bbL)\leq(\sum_{j,i,t}\ld(t)x_{ij,t})/m=(\sum_j\LC_j)/m$, which shows that the
expected latency cost incurred for clients $j$ with $\tau_j\leq t_0$ is at most 
$\sum_j\LC_j$.

\begin{corollary} \label{lpnormcor}
One can obtain an $O\bigl(p\log n\max\{\log n,\log m\}\bigr)$-approximation algorithm for
$\Lf_p$-\mlufl.
\end{corollary}
\begin{proof}
As is standard, we enumerate all possible values of the $\Lf_p$-norm of the optimal
client latencies in powers of $2$, losing potentially another factor of
$2$ in the approximation factor. To avoid getting into issues about estimating the
solution-cost for a given guess (since our algorithms are randomized), we proceed as
follows. For a given guess $\Lat$, we solve \eqref{mlufllp} modifying the objective to be   
$\min\ \ \sum_{i,t} f_iy_{i,t}+\sum_{j,i,t} c_{ij}x_{ij,t}$, and
we adding the constraint $\sum_{j,i,t}t^px_{ij,t}\leq \Lat^p$. 
Among all such guesses and corresponding optimal solutions, let $(x,y,z)$ be the solution 
that minimizes $\sum_{i,t} f_iy_{i,t}+\sum_{j,i,t} c_{ij}x_{ij,t}+\Lat$.
Let $\OPT$ denote this minimum value. Note that $\OPT\leq 2\iopt$, where $\iopt$ is the
optimum value of the $\Lf_p$-\mlufl instance. 
We apply Theorem~\ref{pgrowththm} (with $\ld(x)=x^p$) to round $(x,y,z)$.
Let $F$, $C$, and $L=\sum_j L_j$, denote respectively the (random) facility-opening,
connection-, and latency- cost of the resulting solution. 
The bounds in Theorem~\ref{pgrowththm} imply 
that $\E{(\sum_j L_j)^{1/p}}\leq\bigl(\E{\sum_j L_j}\bigr)^{1/p}\leq O(p\log^2 n)\Lat$, 
which combined with the bounds on $\E{F}$ and $\E{C}$, 
shows that the expected total cost is $O\bigl(p\log n\max\{\log n,\log m\}\bigr)\OPT$. 
\end{proof}

We obtain significantly improved guarantees for related \mlufl, metric uniform \mlufl, and
\ml with (convex) latency functions of growth $p$. 
For related \mlufl and \ml, the analyses in Sections~\ref{mlufl-related} and~\ref{lpml}
directly yield an $O\bigl(2^{O(p)}\bigr)$-approximation guarantee since 
for both problems, we can (deterministically) bound the delay of client $j$ by
$O\bigl(\text{$\al$-point of $j$}\bigr)$ (for suitable $\al$).
For metric uniform \mlufl, we obtain an $O(1)$-approximation bound 
as a consequence of Theorem~\ref{metricredn}: this follows because one can devise an
$O(1)$-approximation algorithm for the zero-facility-cost version of the problem by adapting
the ideas used in~\cite{BaevRS09}.
Thus, we obtain an $O(1)$-approximation for the $\Lf_p$-versions of related-\mlufl,
metric uniform \mlufl, and \ml.  

\vspace{-2ex}
\paragraph{$k$-route \mlufl with length bounds.}
All our algorithms easily generalize to 
{\em $k$-route length-bounded \mlufl}, where 
we are given a budget $B$ and we may use (at most) $k$ paths starting at $r$ of
($d$-) length at most $B$ to traverse the open facilities and activate them. 
This captures the scenario where one can use $k$ vehicles in parallel, each with capacity
$B$, starting at the root depot to activate the open facilities. Observe that with
$B=\infty$, we obtain a generalization of the {\em $k$-traveling repairmen} problem
considered in~\cite{FHR}. 

We modify \eqref{mlufllp} by setting $\Time=B$ and setting the RHS of \eqref{eq:4} to
$kt$. In Algorithm~\ref{genalg}, we now obtain a tour $\Tour_\ell$ in phase $\ell$ of
(expected length) $O(\log^2 n)kt_\ell$. Each facility $i\in\Tour_\ell$ satisfies
$\sum_{t\leq t_\ell} y_{i,t}>0$, so $d(i,r)\leq t_{\ell}$, and we may therefore divide
$\Tour_\ell$ into $k$ tours of length at most $O(\log^2n)t_\ell$. Thus, we obtain the
same guarantee on the expected cost incurred, and we violate the budget by an
$O(\log^2 n)$-factor, that is, we get a bicriteria 
$\bigl(\polylog,O(\log^2 n)\bigr)$-approximation. 
Similarly, for related \mlufl and \ml, 
we obtain an $\bigl(O(1),O(1)\bigr)$-approximation. 
For \ml, we may again use either \eqref{mllp1} or \eqref{mllp2}: in both LPs we set
$\Time=B$; in \eqref{mllp1}, we now have the constraint $\sum_e d_ez_{e,t}\leq kt$, and in
\eqref{mllp2}, the RHS of \eqref{onep} is now $k$.
For metric uniform \mlufl, we modify \eqref{uniflp} in the obvious way: we now have
$\sum_t y_{i,t}\leq k$ for each time $t$, and $t$ now ranges from 1 to $B$. We can again
apply Theorem~\ref{metricredn} here to obtain a ({\em unicriteria}) $O(1)$-approximation
algorithm: 
for the zero-facility-location problem (where we may now ``open'' at most $kB$   
facilities), we can adapt the ideas in~\cite{BaevRS09} to devise an $O(1)$-approximation
algorithm. 

Finally, these guarantees extend to latency functions of bounded growth  
(in the same way that guarantees for \mlufl extend to the setting with latency functions).  
Thus, in particular, we obtain an {\em $O(1)$-approximation algorithm} for the
$\Lf_p$-norm $k$-traveling repairmen problem; this is the {\em first} approximation 
guarantee for this problem. 

\vspace{-2ex}
\paragraph{Non-uniform latency costs.}
We consider here the setting where each client $j$ has a (possibly different) time-to-cost
conversion factor $\ld_j$, which measures $j$'s sensitivity to time delay (vs. connection
cost); so the latency cost of a client $j$ is now given by $\ld_jt_j$, where $t_j$ is the
delay faced by the facility serving $j$.  

All our guarantees in Sections~\ref{mlufl-related},~\ref{mlufl-unif}, and~\ref{lpml}
continue to hold in this non-uniform latency setting. 
In particular, we obtain a constant approximation guarantee for related metric \mlufl,
metric uniform \mlufl, and \ml. 
Notice that the metric uniform \mlufl problem cannot now be solved via a reduction to the
metric-\ufl variant discussed in Section~\ref{mlufl-unif}; however we can still use
Corollary~\ref{redncorr} to obtain a 10.773-approximation. For general \mlufl, it is not
hard to see that our analysis goes through under the the assumption $\Time=\poly(m)$. However, 
the scaling trick used to bypass this assumption leads to an extra 
$O\bigl(\log(\frac{\lambda_{\max}}{\lambda_{\min}})\bigr)$ factor in the approximation.

Recall that the scaling factor $\bbL$ (in the definition of $t_\ell$)
in Section~\ref{mlufl-gen} was defined as $\sum_j \bL_j/m$, 
where $\bL_j = \sum_{i,t} tx_{ij,t}$. Now, $\bL_j = \sum_{i,t} \lambda_jtx_{ij,t}$,
and we set $\bbL = \sum_j \bL_j/\sum_j \lambda_j$. 
One can again argue that the expected latency-cost of each client $j$ is at most 
$\ld_j\cdot O(\log^2 n)\cdot\max\bigl\{\bbL,\tau_j\}$, so we incur an $O(\log^2 n)$-factor in
the latency-cost.
The number of phases, however, is
$\mathcal{N} := \ceil{\log_2(2\tau_{\max}/\bbL) + 4\log_2 m}$, and $2\tau_{\max}/\bbL
= O\bigl(\frac{\tau_{\max}\sum_j\lambda_j}{\sum_j \bL_j}) \le  
O\bigl(\frac{m\lambda_{\max}}{\lambda_{\min}})$, which gives an extra
$\log_2(\lambda_{\max}/\lambda_{\min})$ \nolinebreak \mbox{factor in the facility-opening cost.}

Finally, as before, these guarantees also translate to the $k$-route length-bounded
versions of our problems.

\appendix

\section{Proofs omitted from the main body} \label{append-gen} \label{append-lpml}

\begin{proofof}{Theorem~\ref{lbthm}}
Consider a \gst instance $\bigl(H=(V,E),r,\{d_e\}_{e\in E},\{G_j\sse V\}_{j=1}^m)$. 
Let $n=|V|$.  
We may assume that $H$ is the complete graph, $d$ is a metric, and the groups are
disjoint. We abbreivate $\rho_{n,m}$ to $\rho$ below. 
Let $\Qc$ denote the collection of all solutions to the path-variant of \gst; that is,
$\Qc$ consists of all paths starting at $r$ that visit at least one node of each
group. For a path $Q\in\Qc$ and a group $G_j$, define $Q_j$ to be the portion of $Q$ from
$r$ to the first node of $G_j$ lying on $Q$. Let $d_j(Q)=\sum_{e\in Q_j}d_e$ be the length
of $Q_j$; that is, $d_j(Q)$ is the latency of group $j$ along path $Q$. 
Consider the following LP-relaxation for the path-variant of \gst, and its dual. 
We have a variable $x_Q$ for every $Q\in\Qc$ indicating if path $Q$ is chosen. 
We use $Q$ below to index the paths in $\Qc$.

\vspace{-5pt}
\noindent \hspace*{-6ex}
\begin{minipage}[t]{.49\textwidth}
\begin{alignat*}{3}
\min & \quad & M & \tag{P'} \label{gsp-p} \\
\text{s.t.} && \sum_Q d_j(Q)x_Q & \leq M \qquad && \forall j \\[-1ex]
&& \sum_Q x_Q & \geq 1 \\[-1ex]
&& x & \geq 0.
\end{alignat*}
\end{minipage}
\quad \rule[-19ex]{1pt}{16ex}
\begin{minipage}[t]{.49\textwidth}
\begin{alignat}{3}
\max & \quad & \al & \tag{D'} \label{gsp-d} \\
\text{s.t.} && \sum_j \ld_jd_j(Q) & \geq \al \qquad && \forall Q \label{mglineq} \\[-1ex]
&& \sum_j\ld_j & \leq 1 \notag \\[-1ex]
&& \al,\ld & \geq 0. \notag
\end{alignat} 
\end{minipage}

\medskip 
\noindent
\eqref{gsp-p} has an exponential number of variables. But observe that separating over the
constraints \eqref{mglineq} in the dual involves solving an \mgl problem. The minimum
value (over all $Q\in\Qc$) of the LHS of \eqref{mglineq} is the optimal value of the \mgl
problem defined by $\{(G_j,\ld_j)\}$, where we seek to minimize the weighted sum of
client latency costs. (This weighted group latency problem 
can be reduced to the unweighted problem by ``creating'' $\ld_j$ copies of each group
$G_j$ (we can scale the $\ld_j$s so that they are integral); equivalently (instead of
explicitly creating copies), one can simulate this copying-process in whatever algorithm
one uses for (unweighted) \mgl.)    
Thus, a $\rho$-approximation algorithm for \mgl yields a $\rho$-approximate separation
oracle for \eqref{gsp-d}. Now, applying an argument similar to the one used by Jain et
al.~\cite{JainMS03} shows that one can use this to (also) obtain a $\rho$-approximate
solution $(x,M)$ to \eqref{gsp-p}.   

We now use randomized rounding to round $(x,M)$ and obtain a group Steiner tree of cost at 
most $O(\log m)M$. We pick path $Q$ independently with probability 
$\min\{4\log m\cdot x_Q,1\}$. Let $\Qc'\sse\Qc$ denote the collection of paths
picked. Note that for every group $G_j$, we have 
$\sum_{Q\in\Qc:d_j(Q)\leq 2M}x_Q\geq\frac{1}{2}$. So a standard
set-cover argument shows that with probability at least $1-1/m$, for every $j$, 
there is some path $Q^\br j\in\Qc'$ such that $d_j(Q^\br j)\leq 2M$. 
We may assume that $\sum_Q x_Q=1$, so Chernoff bounds show that $|\Qc'|=O(\log m)$ with
overwhelming probability. 
The group Steiner tree $\T$ consists of the union of all the $Q^\br j_j$ (sub)paths
(deleting edges to remove cycles as necessary). Clearly, the cost of $\T$ is at most
$|\Qc'|\cdot 2M=O(\log m)M$. Note that $\T$ also yields a path of length $O(\log m)M$
starting at $r$ and visiting all groups, so the integrality gap of \eqref{gsp-p} is
$O(\log m)$.
\end{proofof}

\paragraph{Extension of Lemma~\ref{apxsep} to \textsf{MGL}.}
Notice that we did not use anything specific to the minimum-latency problem
in the proof, and so essentially the same proof also applies to \mgl. 
Recall that in \mgl, we have a set $\F$ of facilities and a $d$-metric on
$\F\cup\{r\}$, and a collection of $m$ groups $\{G_j\sse\F\}$. 
Analogous to $\mlpath{a,b}$, the LP-relaxation with path variables for \mgl and its dual
are as follows. 

\newcommand{\mglpath}[1]{\ensuremath{\bigl(\text{LP'}_{\Pcol}^{({#1})}\bigr)}}
\newcommand{\mgltree}[1]{\ensuremath{\bigl(\text{LP'}_{\Tcol}^{({#1})}\bigr)}}

\vspace{-5pt}
\noindent \hspace*{-3ex}
\begin{minipage}[t]{.47\textwidth}
\begin{alignat*}{3}
\min & \quad & \sum_{j,t}tx_{j,t} & \tag{LP'$_{\Pcol}^{(a,b)}$} \label{mgllp} \\
\text{s.t.} && \sum_{t} x_{j,t} & \geq 1 \quad\ && \forall j \\ 
&& \sum_{P\in\Pcol_{bt}} z_{P,t} & \leq a \quad\ && \forall t \\ 
&& \sum_{P\in\Pcol_{bt}: G_j\cap P\neq\es}z_{P,t} & \geq \sum_{t'\leq t}x_{j,t'} 
\quad\ && \forall j,t \\ 
&& x, z & \geq 0. \notag
\end{alignat*}
\end{minipage}
\ \rule[-29ex]{1pt}{26ex}\!\!\!
\begin{minipage}[t]{.54\textwidth}
\begin{alignat}{3}
\max & \quad & \sum_j\al_j & - a\sum_t\beta_t \tag{LD'$_{\Pcol}(a,b)$} \label{mgldlp} \\
\text{s.t.} && \al_j & \leq t+\sum_{t'\geq t}\tht_{j,t'} \quad\ && \forall j,t \notag \\
&& \sum_{j:G_j\cap P\neq\es}\tht_{j,t} & \leq \beta_t \quad\ && \forall t, P\in\Pcol_{bt}
\label{dmglineq} \\ 
&& \al, \beta, \tht & \geq 0. \notag 
\end{alignat}
\end{minipage}

\medskip \noindent
The LP-relaxation $\mgltree{a,b}$ with tree variables is obtained by replacing
$\Pcol_{bt}$ with $\Tcol_{bt}$ in \eqref{mgllp}. 
The orienteering problem that we need to solve now to separate over the constraints
\eqref{dmglineq} is {\em group orienteering}: given a reward $\tht_{j,t}$ for each group
$G_j$, we want to determine if there is a path (or tree) rooted at $r$ of length at
most $bt$ such that the total reward of the groups covered by it is more than $\beta_t$.
Given these changes, the proof that one can obtain a feasible solution $(x,y)$ to
$\mgltree{a,b}$ of cost at most the optimal value of $\mglpath{1,1}$ is as in the
proof of Lemma~\ref{apxsep}. 

\end{document}